\documentclass[runningheads]{llncs}

\usepackage[mobile]{eccv}

\usepackage{eccvabbrv}

\usepackage{graphicx}
\usepackage{booktabs}
\usepackage{multirow}
\usepackage{threeparttable}
\usepackage{pifont}            
\usepackage[normalem]{ulem}    
\usepackage{dsfont}            
\usepackage{epigraph}
\usepackage{algorithm}
\usepackage{algorithmic}
\usepackage{listings}

\usepackage[accsupp]{axessibility}  

\usepackage{hyperref}
\hypersetup{
        colorlinks,
        linkcolor={Maroon},
        citecolor={MidnightBlue},
        urlcolor={teal!85!black}
}
\newcommand{\appref}[1]{Appendix~\ref{#1}}

\usepackage{orcidlink}

\DeclareMathOperator*{\argmax}{arg\,max}
\spnewtheorem{assumption}{Assumption}{\bfseries}{\itshape}

\newcommand{\cmark}{\textcolor{green!60!black}{\ding{51}}}  
\newcommand{\xmark}{\textcolor{red!70!black}{\ding{55}}}    
\newcommand{\red}[1]{#1}
\newcommand{\rev}[1]{#1}
\newcommand\samethanks[1][\value{footnote}]{\footnotemark[#1]}

\begin{document}

\title{\textit{Whence the Voice?} \\Self-supervised Dual-source
  Audio-Visual Localisation via Selective Convergence}

\titlerunning{Whence the Voice?}

\author{Han Hu\thanks{Equal contribution.}\orcidlink{0009-0001-3731-0067} \and
        Dongheng Lin\samethanks\orcidlink{0009-0004-5834-9077} \and
        Yuqi Hou\orcidlink{0009-0002-3086-8698} \and
        Haotian Li\orcidlink{0009-0003-6359-2015} \and
        Hyung Jin Chang\orcidlink{0000-0001-7495-9677} \and
        Jianbo Jiao\orcidlink{0000-0003-0833-5115}}

\authorrunning{H.~Hu et al.}

\institute{The \href{https://mix.jianbojiao.com/}{MIx Group}, School of Computer Science, University of Birmingham, UK\\
  Project page: \url{https://happy-new-bears.github.io/scav-project-page/}}

\maketitle
\setcounter{footnote}{0}  

\begin{abstract}
Localising multiple sound sources in visual scenes remains a fundamental challenge in multimodal perception due to an inherent circular dependency: separating mixed audio requires knowing source locations, while identifying sound-producing regions requires separated audio signals.
In this paper, we \red{focus on the dual-source setting and} discover a {selective convergence} in self-supervised audio-visual learning: when presented with multiple sound sources, contrastive models naturally converge to the most salient audio-visual correspondence rather than attempting to represent all sources equally.
This emergent phenomenon, analogous to human selective auditory attention, enables us to break the above circular dependency through a progressive two-stage framework: first, leveraging selective convergence to identify dominant sources, and then exploiting these learned priors to uncover remaining sources.
Our self-supervised approach achieves \red{the best performance among self-supervised methods} on \rev{dual-source} benchmarks without requiring any manual annotations, \red{and even surpasses} some weakly-supervised approaches \red{on certain metrics}.
Furthermore, we identify a fundamental evaluation inconsistency in existing benchmarks: comparing continuous localisation heatmaps against bounding-box annotations creates systematic biases, particularly for non-axis-aligned objects where the bounding box includes substantial background regions.
To address this, we introduce pixel-level segmentation masks to the existing benchmark, enabling spatially-aligned evaluation.
Together, these results suggest that embracing rather than suppressing selectivity offers a scalable, annotation-free route to multi-source localisation.


\epigraph{\makebox[1\textwidth][r]{``They rose expectant: eye and ear waited...''}}{\textit{Charlotte Brontë, Jane Eyre}}

\end{abstract}

\section{Introduction}
\label{sec:intro}
\begin{figure}[t]
      \centering
      \includegraphics[width=.99\linewidth]{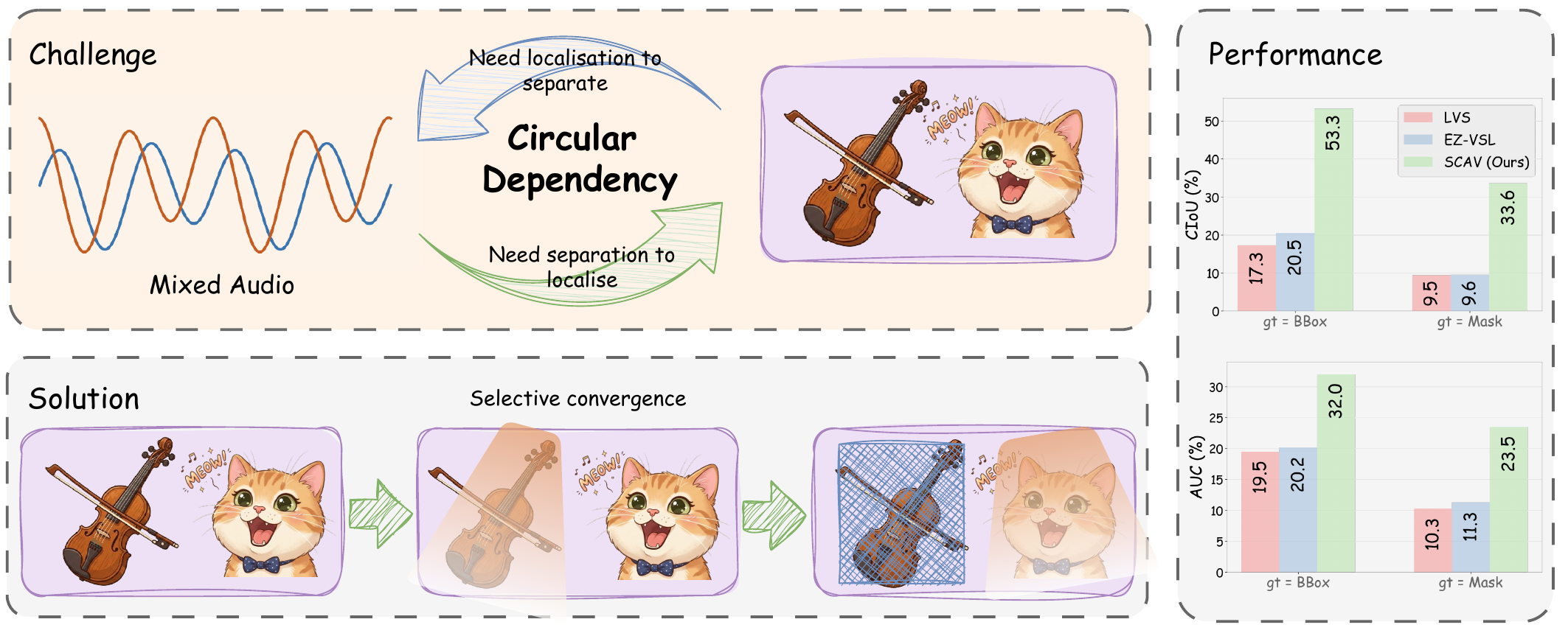}
      \caption{\textbf{Top-left}: The fundamental paradox where visual localisation requires separated audio, while audio separation needs spatial visual information, creating an intractable circular dependency. \textbf{Bottom-left}: Our SCAV framework breaks this dependency by leveraging selective convergence. The model naturally focuses on one dominant source at first, then progressively reveals the other source. \textbf{Right}: Performance of the proposed SCAV surpasses other \red{self-supervised} methods across various metrics.}
      \label{fig:teaser}
\end{figure}
Sound rarely occurs in isolation. In natural environments, multiple sources produce sounds simultaneously, creating rich acoustic scenes that humans navigate effortlessly by matching sounds to their visual origins. Such audio-visual correspondence forms the basis of how we understand and interact with the world.
For machines, the challenge is concrete: \textit{given a mixed audio signal and a visual scene with corresponding sounding objects, how can we determine which sounds come from which objects, when neither modality provides clear separation cues?}

Localising multiple simultaneous sound sources requires mixed-signal disentanglement in two different domains.
In the audio domain, sounds mix through direct linear superposition;
In contrast, the visual domain combines objects through spatial arrangement: different sound sources occupy distinct regions in an image.
This domain asymmetry creates a challenging chicken-and-egg paradox: to decompose the additively mixed audio, we need spatial guidance from the visual domain to indicate where each source is located; yet to identify which spatial regions produce sound, we need separated audio signals to establish clear correspondences.
This circular dependency makes the joint localisation task a severely ill-posed inverse problem: multiple valid decompositions exist for both the mixed audio and the visual scene, and without additional constraints, standard optimisation cannot determine which solution corresponds to the true source configuration.

While this circular dependency appears fundamentally intractable, human auditory perception offers an intriguing perspective.
Humans effortlessly navigate cocktail parties through selective attention~\cite{bronkhorst2015cocktail, sussman2017, brattico2017global}: they focus on the most salient speaker first and then shift to others.
This cognitive strategy suggests that perfect simultaneous separation may be neither necessary nor optimal.
When we investigated whether neural networks exhibit similar behaviours in multi-source scenarios, extensive experiments with contrastive audio-visual learning revealed a remarkable phenomenon:
Rather than struggle to represent all sources equally, networks consistently exhibit \emph{selective convergence} and gravitate toward the most dominant audio-visual correspondence.
We formalise this insight into a two-stage progressive framework that leverages selective convergence as a core mechanism \red{for dual-source localisation}.
The first stage exploits this natural bias to identify dominant sources through contrastive learning, while the second stage uses these learned priors with targeted masking to progressively unmask subdominant sources.
This approach effectively converts the ill-posed joint problem into a sequence of well-conditioned optimisations requiring no manual annotations.

To sum up, our contributions are threefold. \textbf{First}, we identify the \emph{selective convergence} phenomenon: when learning from mixed audio, contrastive models naturally focus on the more dominant source rather than all sources equally.
\textbf{Second}, we propose a two-stage progressive framework that leverages selective convergence to break the circular dependency problem.
The first stage uses selective convergence to find dominant sources, while the second stage progressively unmasks remaining sources through targeted masking.
\textbf{Third}, our approach not only achieves promising performance improvements on existing benchmarks, but the stable performance improvement brought by our method is also further validated under a newly introduced, more rigorous evaluation protocol.

\section{Related Works}
\label{sec:related_work}

\rev{Audio-visual sounding source localisation aims to determine the spatial location of sounding objects within a video frame using the accompanying audio signal. The inherent co-occurrence between the auditory and visual modalities provides a powerful self-supervised signal for this task~\cite{Arandjelovic2017LookLL, Arandjelovic2018ObjectsTS, Owens2018AudioVisualSA, Owens2016AmbientSP, Aytar2016SoundNetLS, hu2025audio, um2025object}.}

\rev{Early efforts in audio-visual sound source localisation primarily focused on only one sound source per scene~\cite{senocak2018learning, senocak2020learning}.
Arandjelovic and Zisserman~\cite{Arandjelovic2018ObjectsTS} and Owens and Efros~\cite{Owens2018AudioVisualSA} employed mechanisms akin to Class Activation Map to measure the similarity between image regions and audio features.
More recent approaches adopted contrastive learning frameworks~\cite{oord2018representation} to improve discriminability at the region level.
Chen~\etal~\cite{Chen2021LVS} introduced a contrastive learning framework with a differentiable thresholding mechanism to automatically mine hard negative image fragments.
Mo and Morgado~\cite{Mo2022EZVSL} proposed EZ-VSL, a multiple-instance contrastive learning strategy that aligns audio and visual spaces by maximising the similarity between the audio and the most responsive region in the corresponding image.
Recent efforts have addressed practical challenges such as false negatives~\cite{sun2023learning}, off-screen sounds and background noise~\cite{liu2022visual, choi2025whats}, learning from silence~\cite{juanola2025ssl}, and feature decomposition through slot attention~\cite{kim2025jsa}.
Park~\etal~\cite{park2024clipssl} explored leveraging CLIP's pre-trained visual representations for sound source localisation, demonstrating that vision-language models can provide useful semantic priors for audio-visual correspondence.
However, these methods fail when multiple sources sound simultaneously, which is the common case in real-world scenarios.}

\rev{When multiple sound sources are present in a visual scene, separating their mixed audio requires knowing where each source is located, but finding these locations requires having the separated audio to establish correspondences, which creates a fundamental circular dependency.
To break this dependency, researchers have explored different strategies.
Mix-and-Localize~\cite{mix_and_localize, mo2024multiscale} proposed a self-supervised random walk on an image-sound graph to force source disentanglement through cycle consistency.
Other methods introduce explicit semantic guidance.
The self-supervised approach DSOL~\cite{dsol2020} tackles the paradox by generating a pseudo-class dictionary of visual representations via unsupervised clustering in a single-source setting.
More recently, Kim~\etal~\cite{kim2024learning} introduced an iterative curriculum learning framework that progressively discovers multiple sound sources through object-aware contrastive learning and negative exclusion, though it requires multiple hand-tuned thresholds.
The weakly-supervised method AVGN~\cite{mo2023audiovisual} directly uses true video-level class labels to guide learnable class tokens for explicit source grouping and disentanglement, while Dual Mean-Teacher~\cite{guo2023dual} employs semi-supervised learning with dual teacher-student structures to generate high-quality pseudo-labels.
More recent work introduces external modalities: T-VSL~\cite{mahmud2024tvsl} leverages text representations from AudioCLIP's tri-modal embedding space to guide category-specific feature extraction, while OA-SSL~\cite{um2025object} employs Multimodal Large Language Models to generate contextual descriptions that distinguish actively sounding objects from silent ones.
In contrast to methods that engineer complex mechanisms, rely on weakly-supervised labels, or introduce external text modalities, our work leverages the emergent property of contrastive learning to progressively break the circular dependency in a fully self-supervised manner \red{with} only audio-visual data.}

\section{Method}
\label{sec:method}
\subsection{Preliminaries}
\label{sec:preliminaries}
Given a video containing two simultaneously sounding sources, our goal is to localise both sound sources in the visual scene without manual annotations.
Let $\mathcal{X} = (I, \mathbf{a}_{\text{mix}})$ denote an audio-visual pair, where the visual frame $I \in \mathbb{R}^{H \times W \times 3}$ captures a scene with two sound-producing objects, and the mixed audio spectrogram $\mathbf{A}_{\text{mix}} \in \mathbb{R}^{T \times F}$ contains $T$ time steps and $F$ frequency bins.
Our objective is to produce two continuous-valued spatial localisation maps $\mathbf{M}_1, \mathbf{M}_2 \in [0,1]^{H \times W}$, where each map indicates the location of the corresponding sound source.

Localising sound sources requires computing audio-visual correspondences between $\mathbf{a}_i$ and $I$ to obtain $\mathbf{M}_i$, yet we only observe the mixed audio $\mathbf{a}_{\text{mix}} = \sum_i \mathbf{a}_i$ where individual source contributions $\mathbf{a}_i$ are unknown.
One might attempt to first separate $\mathbf{a}_i$ from $\mathbf{a}_{\text{mix}}$, but separation requires knowing $\mathbf{M}_i$ to assign frequency components to spatial locations.
This creates a \emph{chicken-egg paradox}: localisation needs separated audio, while separation needs spatial maps.
Without additional constraints, this circular dependency makes the problem ill-posed.

While this circular dependency appears intractable, we observe that contrastive learning naturally exhibits \emph{selective convergence}: rather than representing all sources equally, models gravitate toward \red{a single, most salient} audio-visual correspondence.
We propose the Selective Convergence Model for Audio-Visual localisation (SCAV), a two-stage framework that leverages this phenomenon.
\red{Through contrastive learning, the first stage identifies the source that the model is biased toward, which we term the \emph{dominant source} $\mathbf{M}_{\text{dom}}$; the other is the \emph{subdominant source} $\mathbf{M}_{\text{sub}}$. This dominance is determined by the model's bias, instead of any predefined property of the audio or visual input. Importantly, our framework relies only on the fact that selective convergence isolates one source as the spatial prior; it is agnostic to which of the two sources becomes dominant.}
The second stage uses $\mathbf{M}_{\text{dom}}$  as the spatial prior to progressively reveal the subdominant source $\mathbf{M}_{\text{sub}}$.
This approach effectively converts the ill-posed joint problem into a sequence of well-conditioned optimisations in a fully self-supervised manner.

\begin{figure*}[t]
      \centering
      \includegraphics[width=\textwidth]{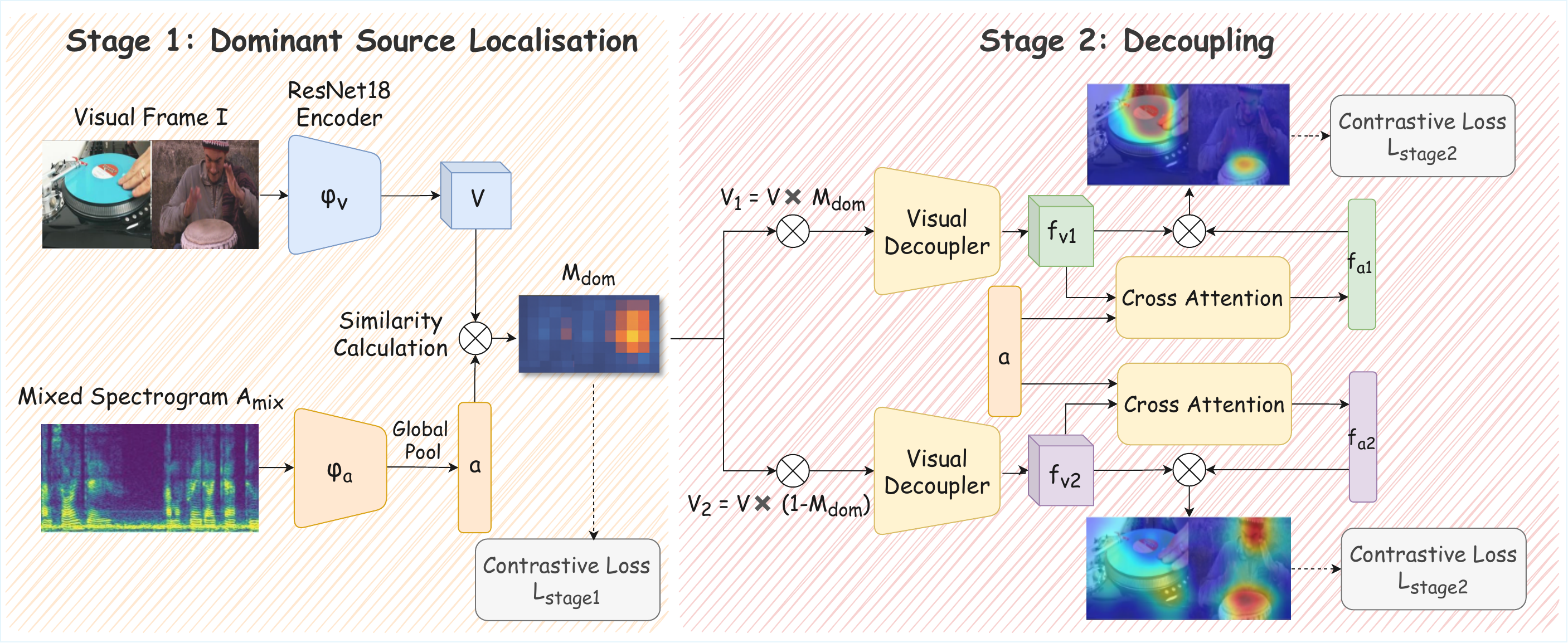}
      \caption{\textbf{Architecture of our Selective Convergence Audio-Visual localisation (SCAV) framework.} \textbf{Stage 1 (left)}:
      We extract visual features $V$ using a frozen ResNet-18 encoder. The visual features are aligned with audio features $\mathbf{a}$ via contrastive learning to generate a coarse localisation heatmap $\mathbf{M}_{\text{dom}}$ that identifies the more dominant sounding source.
      \textbf{Stage 2 (right)}: Using the Stage 1 heatmap as a spatial prior, visual and audio decouplers progressively separate the mixed features into source-specific representations that yield refined localisation maps for both sources.
      }
      \label{fig:method}
\end{figure*}

\subsection{Dominant Source Localisation via Selective Convergence}
\label{subsec:stage1}
We first extract features from different modalities:
\begin{equation}
V = \Phi_v(I) \in \mathbb{R}^{D \times h \times w} , \quad
\mathbf{a} = \text{GAP}(\Phi_a(\mathbf{A}_{\text{mix}})) \in \mathbb{R}^{D}
\end{equation}
where $h = H/32$, $w = W/32$,   and GAP denotes global average pooling.
During training, we initialise the visual encoder $\Phi_v$ with ImageNet pretrained weights and keep it frozen to leverage robust visual representations, while the audio encoder $\Phi_a$ is trained from scratch to learn audio-visual correspondence

For audio features $\mathbf{a}_i$ from sample $i$ and visual features $V_j$ from sample $j$, we calculate the dot product between $\mathbf{a}_i$ and the feature vector at each spatial patch location $(h, w)$ of $V_j$:
\begin{equation}
\mathcal{S}_{i \rightarrow j} = \langle V_j, \mathbf{a}_i \rangle \in \mathbb{R}^{h \times w}
\end{equation}
where each element $\mathcal{S}_{i \rightarrow j}^{(h,w)}$ represents the similarity score between the audio and the visual features at position $(h,w)$.

To train the encoders in a self-supervised manner, we adopt the contrastive learning framework from~\cite{Chen2021LVS} with differentiable thresholding to automatically identify positive and negative spatial regions.
The soft masks $\hat{m}_p$ and $\hat{m}_n$ are computed via differentiable thresholding:
\begin{align}
\hat{m}_p &= \sigma((\mathcal{S}_{i \rightarrow i} - \epsilon_p)/\tau), \quad
\hat{m}_n = 1 - \sigma((\mathcal{S}_{i \rightarrow i} - \epsilon_n)/\tau) \label{eq:masks}
\end{align}
with $\epsilon_p = 0.65$, $\epsilon_n = 0.4$, temperature $\tau = 0.03$ and $\sigma$ means  Sigmoid function.

These masks partition spatial locations into positive pairs (regions with strong audio-visual correspondence, $\hat{m}_p$) and negative pairs (regions with weak correspondence, $\hat{m}_n$).
The positive and negative scores are then computed as:
\begin{equation}
  P_i = \frac{1}{|\hat{m}_p|} \langle \hat{m}_p, \mathcal{S}_{i \rightarrow i} \rangle,
  \label{eq:positive_score}
  \end{equation}
  \begin{equation}
  N_i = \frac{1}{|\hat{m}_n|} \langle \hat{m}_n, \mathcal{S}_{i \rightarrow i} \rangle + \frac{1}{hw} \sum_{j \neq i} \langle \mathbf{1}, \mathcal{S}_{i \rightarrow j} \rangle
  \label{eq:negative_score}
\end{equation}
where $P_i$ aggregates similarity over positive regions (high audio-visual correspondence) and $N_i$ aggregates over negative regions (low correspondence within sample $i$, plus cross-sample negatives from $j \neq i$).

The contrastive loss is:
\begin{equation}
\mathcal{L}_{\text{stage1}} = -\frac{1}{B} \sum_{i=1}^{B} \log \frac{\exp(P_i)}{\exp(P_i) + \exp(N_i)}.
\label{eq:stage1_loss}
\end{equation}

\rev{Training with such contrastive loss on mixed audio leads the model to exhibit \emph{selective convergence}: the learned similarity map $\mathcal{S}_i = \langle V_i, \mathbf{a}_i \rangle \in \mathbb{R}^{h \times w}$ naturally concentrates on the dominant source with the most obvious audio-visual correspondence.}
\rev{This arises from the simplicity bias of gradient-based optimisation~\cite{xue2023icml}.}

\rev{To extract a clean spatial prior $\mathbf{M}_{\text{dom}}$ from this similarity map, we further apply a simple post-processing step (detailed in \cref{alg:dominant_extraction}).}

\subsection{Progressive Multi-source Localisation with Spatial Priors}
\label{subsec:spatial_prior}
Having identified the dominant source $\mathbf{M}_{\text{dom}}$ in Stage 1, we now leverage it as a spatial prior
to break the circular dependency.
The key insight is simple: with $\mathbf{M}_{\text{dom}}$ telling us where the dominant source is located, we can (1) partition visual features into source-specific regions, (2) use these partitioned visual features to guide audio separation via cross-attention, and (3) localise each source independently with the decoupled audio-visual pairs.

\paragraph{\textbf{Feature extraction and projection.}}
We reuse the encoders from the first stage to extract visual features $V \in \mathbb{R}^{512 \times h \times w}$ and audio features $\mathbf{a} \in \mathbb{R}^{512}$.
To enable effective decoupling, we project these features into a unified embedding space and add 2D positional
encoding to preserve spatial structure:
\begin{equation}
  \tilde{V} = \text{Proj}_v(V) + \text{PE}_{2D} \in \mathbb{R}^{D \times h \times w},
\end{equation}
\begin{equation}
  \tilde{\mathbf{a}} = \text{Proj}_a(\mathbf{a}) \in \mathbb{R}^{D}
\end{equation}
where $\text{Proj}_v$ and $\text{Proj}_a$ are learnable linear  projection layers, and $\text{PE}_{2D}$ denotes 2D sinusoidal positional encoding.

\paragraph{\textbf{Visual feature decoupling.}}
With the spatial prior $\mathbf{M}_{\text{dom}}$, we partition the visual features into two complementary regions:
\begin{align}
V_1 &= \tilde{V} \odot \mathbf{M}_{\text{dom}}, \quad
V_2 = \tilde{V} \odot (1 - \mathbf{M}_{\text{dom}})
\end{align}
where $V_1$ captures the dominant source region and $V_2$ captures the complementary region (other possible source region and the background).

A visual decoupler $\Phi_v^{\text{dec}}$ then refines these coarsely separated features to produce more discriminative source-specific visual features:
\begin{equation}
f_{v1}, f_{v2} = \Phi_v^{\text{dec}}(V_1, V_2), \quad f_{vi} \in \mathbb{R}^{D \times h \times w}.
\end{equation}

\paragraph{\textbf{Audio feature decoupling.}}
\label{subsec:stage2}
\rev{With decoupled visual features $(f_{v1}, f_{v2})$ ready, we can now separate the mixed audio feature $\tilde{\mathbf{a}}$ into source-specific representations via the audio decoupler $\Phi_a^{\text{dec}}$, shown as the cross-attention branch in Stage 2 (right) of \cref{fig:method}.}
\rev{The key idea is to use each visual feature as a spatial guide: the audio decoupler $\Phi_a^{\text{dec}}$ employs
cross-attention where the mixed audio serves as the query, and each visual feature provides keys and values:}
\begin{equation}
f_{a1}, f_{a2} = \Phi_a^{\text{dec}}(\tilde{\mathbf{a}}, f_{v1}, f_{v2}) \in \mathbb{R}^{D}.
\end{equation}

\rev{This allows each audio feature $f_{ai}$ to attend to its corresponding visual region, effectively extracting the audio component associated with that spatial location.}

\paragraph{\textbf{Source-specific contrastive learning.}}
After having decoupled both modalities into $(f_{v1}, f_{v2})$ and $(f_{a1}, f_{a2})$, we train each source pair independently using contrastive learning.
For each source $i \in \{1, 2\}$, we compute the similarity map $\mathcal{S}_i = \langle f_{vi}, f_{ai}\rangle$ and optimie:
\begin{equation}
\mathcal{L}_i = -\log \frac{\exp(P_i)}{\exp(P_i) + \exp(N'_i)}
\end{equation}
where $P_i$ and $N'_i$ are positive and negative scores computed via differentiable thresholding
(\cref{eq:masks}).
\rev{Crucially, unlike Stage 1, we exclude cross-sample negative terms in $N'_i$ because the decoupling task focuses on within-sample source assignment rather than cross-sample discrimination.}
The total loss is:
\begin{equation}
\mathcal{L}_{\text{Stage2}} = \mathcal{L}_1 + \mathcal{L}_2.
\end{equation}

\paragraph{\textbf{Inference stage.}}
\rev{At inference, the trained model produces decoupled audio-visual pairs $(f_{v1}, f_{a1})$ and $(f_{v2}, f_{a2})$.
The final localisation maps are obtained by computing similarity and upsampling:}
\begin{equation}
\mathbf{M}_k = \text{Upsample}(\langle f_{vk}, f_{ak} \rangle), \quad k \in \{1, 2\}
\end{equation}
where $\langle \cdot, \cdot \rangle$ denotes the inner product.

\section{Experiments}
\label{sec:experiments}
\subsection{Experimental Setup}
\label{subsec:setup}
\paragraph{\textbf{Implementation details.}}
For visual inputs $I$, each image is resized to $224 \times 224$ resolution.
For audio inputs, signals are sampled at 22.05 kHz and converted to log-spectrograms $\mathbf{A}_{\text{mix}}$ using Short-Time Fourier Transform with 50ms windows and 25ms hop size for 3-second clips.
The visual encoder $\Phi_v$  employs separate ResNet18~\cite{he2016deep} networks initialised with ImageNet pre-trained weights.
The audio encoder $\Phi_a$ uses another ResNet18 architecture with the first convolutional layer adapted to a single channel.
The visual decoupler $\Phi_v^{\text{dec}}$ employs a 4-layer Transformer encoder with 8 attention heads, while the audio decoupler $\Phi_a^{\text{dec}}$ uses cross-attention mechanisms to produce separated audio representations.

\paragraph{\textbf{Training protocol.}}
We employ a two-stage progressive training strategy.
In the first stage, we train  encoders ($\Phi_v$, $\Phi_a$) using  $\mathcal{L}_{\text{stage1}}$ to learn dominant source localisation through selective convergence.
In the second stage, we train the full architecture using  $\mathcal{L}_{\text{stage2}}$, with spatial priors $\mathbf{M}_{\text{dom}}$ provided by the trained first-stage model.
We train our model using the Adam optimizer~\cite{kingma2014adam} with an initial learning rate of $10^{-4}$ and a batch size of 256.

\rev{Note that the two-stage design refers to two \textit{training} stages only; at inference time, the trained model executes a single forward pass without interruption, achieving real-time performance at 43.2 FPS on a single NVIDIA A100 GPU.}

\subsection{Evaluation Metrics}
\label{subsec:metrics}

Following prior works~\cite{mix_and_localize,dsol2020}, we evaluate \rev{dual-source} localisation using CIoU@X, AUC, and CAP.
\rev{Note that for synthetic dual-source benchmarks (\eg VGGSound-Duet), we evaluate each predicted heatmap over the \emph{full} concatenated image domain against zero-padded ground truth masks without any cropping on the predicted heatmap. This ensures that cross-source activations are properly penalised; see \appref{subsec:eval_protocol} for more details.}

(1) CAP (Class-Aware Precision) measures the percentage of correctly localised sources at the class level:
\begin{equation}
CAP = \frac{\sum_{k=1}^{K} \delta_k \mathrm{AP}_k}{\sum_{k=1}^{K} \delta_k}.
\end{equation}

(2) CIoU (Class-Aware IoU) is the intersection-over-union between predicted heatmaps and ground truth:
\begin{equation}
CIoU = \frac{\sum_{k=1}^{K} \delta_k \, IoU_k}{\sum_{k=1}^{K} \delta_k}
\end{equation}
where $\mathrm{IoU}_k$ is calculated based on the predicted sounding object area and the annotated bounding box for the $k$-th class.

(3) AUC is the area under the precision-recall curve:
\begin{equation}
\mathrm{AUC} = \int_{0}^{1} R(t)\,\mathrm{d}t, \quad \text{where} \quad R(t) = \frac{\sum_{k=1}^{K} \delta_k \cdot
   \mathds{1}[\mathrm{IoU}_k \ge t]}{\sum_{k=1}^{K} \delta_k}.
\end{equation}

\subsection{Evaluation on Existing Benchmarks}
\label{subsec:existing_benchmarks}
We first evaluate our approach on three established multi-source localisation benchmarks, \ie MUSIC-Duet, VGGSound-Instruments, and VGGSound-Duet:

\paragraph{\textbf{MUSIC-Duet.}} The dataset originally contains 149 untrimmed duet videos from the MUSIC dataset~\cite{sound_of_pixel}, the accessible set comprises 141 videos that cover 11 classes of musical instruments due to the YouTube availability issue.
The dataset provides bounding-box annotations for dual sources generated using the Faster R-CNN detector by Hu~\etal~\cite{dsol2020}.
Following Hu~\etal~\cite{dsol2020}, the first two videos per instrument category form the test set, and the rest are for training.

As shown in \cref{tab:two_bbox_datasets}, our self-supervised SCAV achieves the best localisation quality (CIoU/AUC) \red{among self-supervised methods}.
However, SCAV exhibits weaker detection performance (CAP) compared to AVGN~\cite{mo2023audiovisual}.

\paragraph{\textbf{VGGSound-Instruments.}}
Hu~\etal~\cite{mix_and_localize} filtered 37 musical instrument categories from VGGSound~\cite{chen2020vggsound} with 32k training videos.
For multi-source evaluation, Hu~\etal~\cite{mix_and_localize} annotated segmentation masks for 446 high-quality frames.
The benchmark synthesises multi-source scenarios by randomly concatenating two frames into $224\times448$ images while summing their audio waveforms.
As shown in \cref{tab:two_bbox_datasets}, the proposed SCAV method achieves the best performance \red{among self-supervised methods} on most metrics (CAP/AUC) and a competitive CIoU@0.1.
\begin{table*}[t]
\centering
\caption{Quantitative results on MUSIC-Duet (BBox) and VGGSound-Duet (BBox).}
\label{tab:two_bbox_datasets}
\resizebox{\textwidth}{!}{
\begin{tabular}{l | c | ccc | ccc}
\toprule
\multirow{2}{*}{\textbf{Method}} &
\multirow{2}{*}{\textbf{Self-Sup.}} &
\multicolumn{3}{c|}{\textbf{MUSIC-Duet (BBox)}} &
\multicolumn{3}{c}{\textbf{VGGSound-Duet (BBox)}} \\
\cmidrule(lr){3-5} \cmidrule(lr){6-8}
& & \textbf{CIoU@0.3(\%)} & \textbf{AUC(\%)} & \textbf{CAP(\%)} &
\textbf{CIoU@0.3(\%)} & \textbf{AUC(\%)} & \textbf{CAP(\%)} \\
\midrule

CoarsetoFine (ECCV'2020)~\cite{qian2020multiple}
& \xmark
& \textcolor{gray}{17.6} & \textcolor{gray}{20.6} & -
& \textcolor{gray}{14.7} & \textcolor{gray}{18.5} & - \\

AVGN (CVPR'2023)~\cite{mo2023audiovisual}
& \xmark
& \textcolor{gray}{32.5} & \textcolor{gray}{24.6} & \textcolor{gray}{50.6}
& \textcolor{gray}{26.2} & \textcolor{gray}{23.8} & \textcolor{gray}{21.9} \\

OA-SSL (CVPR'2025)~\cite{um2025object}
& \xmark
& \textcolor{gray}{45.9} & \textcolor{gray}{36.1} & \textcolor{gray}{64.1}
& \textcolor{gray}{55.2} & \textcolor{gray}{44.8} & \textcolor{gray}{45.9} \\

\midrule
Attention10k (CVPR'2018)~\cite{senocak2018learning}
& \cmark
& 21.6 & 19.6 & -
& 11.5 & 15.2 & - \\

OTS (ECCV'2018)~\cite{Arandjelovic2018ObjectsTS}
& \cmark
& 13.3 & 18.5 & 11.6
& 12.2 & 15.8 & 10.7 \\

DMC (CVPR'2019)~\cite{hu2019deep}
& \cmark
& 17.5 & 21.1 & -
& 13.8 & 17.1 & - \\

DSOL (NeurIPS'2020)~\cite{dsol2020}
& \cmark
& \underline{30.1} & \underline{22.3} & -
& \underline{22.3} & \underline{21.1} & - \\

LVS (CVPR'2021)~\cite{Chen2021LVS}
& \cmark
& 22.5 & 20.9 & -
& 17.3 & 19.5 & - \\

EZ-VSL (ECCV'2022)~\cite{Mo2022EZVSL}
& \cmark
& 24.3 & 21.3 & -
& 20.5 & 20.2 & - \\

Mix-and-Localize (CVPR'2022)~\cite{mix_and_localize}$^*$
& \cmark
& 26.5 & 21.5 & \underline{47.5}
& 21.1 & 20.5 & 16.3 \\

NoPrior (CVPR'2024)~\cite{kim2024learning}
& \cmark
& 2.5$^\dagger$ & 13.2 & 21.7
& 18.1 & 16.8 & \underline{26.8} \\

Slot-Attn (CVPR'2025)~\cite{kim2025jsa}
& \cmark
& 9.9 & 12.7 & 19.6
& 15.5 & 16.2 & 25.0 \\

\textbf{SCAV (Ours)}
& \cmark
& \textbf{47.1} & \textbf{28.6} & \textbf{47.6}
& \textbf{53.3} & \textbf{32.0} & \textbf{53.0} \\

\bottomrule
\multicolumn{8}{l}{\footnotesize $^*$Mix-and-Localise is re-run under our frame-wise evaluation for VGGSound-Duet; OA-SSL and AVGN are taken as originally reported.} \\
\multicolumn{8}{l}{\footnotesize $^\dagger$See \appref{sec:noprior_reproduction} for our reproduction details and analysis.}
\end{tabular}
}
\end{table*}

\begin{table*}[t]
\centering
\caption{Quantitative comparison on VGGSound-Instruments (Mask) and VGGSound-Duet$^{\text{Mask}}$ with segmentation-mask-based evaluation.}
\label{tab:two_mask_datasets_vertical}
\resizebox{\textwidth}{!}{
\begin{tabular}{l | l | c | cccc}
\toprule
\textbf{Dataset} & \textbf{Method} & \textbf{Self-Sup.} &
\textbf{CAP(\%)} & \textbf{CIoU@0.1(\%)} & \textbf{CIoU@0.3(\%)} & \textbf{AUC(\%)} \\
\midrule

\multirow{10}{*}{\textbf{VGGSound-Instruments}}
& CoarsetoFine (ECCV'2020)~\cite{qian2020multiple} & \xmark
& - & \textcolor{gray}{54.2} & - & \textcolor{gray}{12.9} \\
& AVGN (CVPR'2023)~\cite{mo2023audiovisual} & \xmark
& \textcolor{gray}{27.3} & \textcolor{gray}{77.5} & - & \textcolor{gray}{18.2} \\
\cmidrule(lr){2-7}
& Attention10k (CVPR'2018)~\cite{senocak2018learning} & \cmark
& - & 52.3 & - & 11.7 \\
& OTS (ECCV'2018)~\cite{Arandjelovic2018ObjectsTS} & \cmark
& \underline{23.3} & 51.2 & - & 11.2 \\
& DMC (CVPR'2019)~\cite{hu2019deep} & \cmark
& - & 53.7 & - & 12.5 \\
& DSOL (NeurIPS'2020)~\cite{dsol2020} & \cmark
& - & \underline{74.3} & - & \underline{15.9} \\
& LVS (CVPR'2021)~\cite{Chen2021LVS} & \cmark
& - & 57.3 & - & 13.3 \\
& EZ-VSL (ECCV'2022)~\cite{Mo2022EZVSL} & \cmark
& - & 60.2 & - & 14.2 \\
& Mix-and-Localize (CVPR'2022)~\cite{mix_and_localize} & \cmark
& 21.5 & 73.2 & - & 15.6 \\
& \textbf{SCAV (Ours)} & \cmark
& \textbf{32.0} & \textbf{76.0} & - & \textbf{21.1} \\

\midrule

\multirow{6}{*}{\textbf{VGGSound-Duet$^{\text{Mask}}$}}
& \textcolor{gray}{AVGN (CVPR'2023)~\cite{mo2023audiovisual}} & \xmark
& \textcolor{gray}{33.61} & \textcolor{gray}{55.43} & \textcolor{gray}{18.44} & \textcolor{gray}{17.39} \\
\cmidrule(lr){2-7}
&  LVS (CVPR'2021)~\cite{Chen2021LVS} & \cmark
& 18.00 & 31.51 & 9.53 & 10.28 \\
& EZ-VSL (ECCV'2022)~\cite{Mo2022EZVSL} & \cmark
& 16.82 & 37.97 & 9.58 & 11.35 \\
& NoPrior (CVPR'2024)~\cite{kim2024learning} & \cmark
& 20.26 & 37.02 & 5.77 & 10.25 \\
& Slot-Attn (CVPR'2025)~\cite{kim2025jsa} & \cmark
& 15.21 & 30.87 & 3.63 & 8.92 \\
& \textbf{SCAV (Ours)} & \cmark
& \textbf{39.30} & \textbf{69.42} & \textbf{33.59} & \textbf{23.52} \\

\bottomrule
\end{tabular}
}
\end{table*}

\paragraph{\textbf{VGGSound-Duet.}}
To assess scalability, we evaluate on VGGSound-Duet~\cite{mo2023audiovisual}, the largest multi-source benchmark with 5,158 test videos spanning 220 diverse categories beyond musical instruments.
Like VGGSound-Instruments, it synthesises dual-source scenarios but uses bounding box annotations from the VGGSound-Source test set~\cite{Chen2021LVS}.

SCAV achieves \red{the best performance among self-supervised methods} across all metrics in \cref{tab:two_bbox_datasets} and \cref{tab:two_mask_datasets_vertical}, \red{even surpassing} some weakly-supervised methods \red{on certain metrics}.
\red{These consistent gains} on the larger-scale dataset \red{suggest} that our approach particularly benefits from increased data diversity and scale.

\paragraph{\textbf{Important observation.}}
During analysis of the VGGSound-Duet results, we discovered a fundamental flaw in the evaluation of the bounding box that may have been overlooked by the community.
Bounding boxes treat all enclosed pixels equally: both object and background receive the same weight in IoU computation.
This evaluation protocol essentially rewards over-activation rather than precise localisation.
\cref{fig:bbox_issue} vividly demonstrates this issue with several representative examples.
In the \textit{playing flute} case, while our model precisely localises the sound-producing region around the instrument, the ground truth
bounding box includes substantial empty corners due to the diagonal instrument orientation.
Similarly, for \textit{turkey gobbling}, our heatmap identifies the turkey as the sound source, yet the bounding box encompasses large background areas.
These examples reveal a possible misalignment: current box-based metrics paradoxically penalise precise localisation while rewarding models that indiscriminately activate background regions within bounding boxes.

\begin{figure}[t]
      \centering
      \includegraphics[width=\columnwidth]{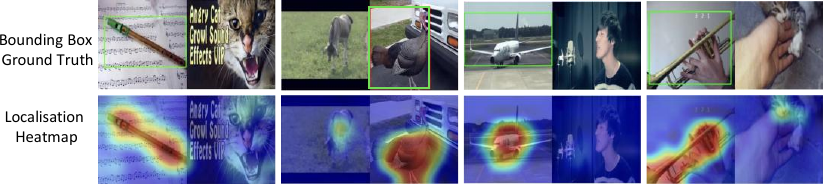}
      \caption{\textbf{Bounding box evaluation can be misleading.} Shown by four representative examples. Top: Ground-truth bounding boxes encompass substantial background regions beyond the actual sound sources. Bottom: Our predicted localisation heatmaps precisely identify sounding regions.}
      \label{fig:bbox_issue}
\end{figure}

Recent work in object detection~\cite{evaluate_right} identified a similar misalignment of the evaluation and demonstrated that segmentation masks provide a more faithful assessment by precisely delineating object boundaries.
Although VGGSound-Instruments~\cite{mix_and_localize} introduced segmentation masks, its limited scale (446 test samples) and restriction to musical instruments prevent comprehensive evaluation across diverse acoustic scenarios.
However, large-scale pixel-precise benchmarks remain scarce, as segmentation annotation was prohibitively expensive before recent advances like SAM~\cite{sam}.
This observation motivates us to establish a rigorous evaluation benchmark for the audio-visual localisation task using segmentation masks at scale, which could enable accurate measurement of what models truly localise across diverse sound categories.

\subsection{A New Evaluation Protocol}
\label{subsec:new_protocol}
\paragraph{\textbf{VGGSound-Duet$^{\text{Mask}}$.}}
To address the evaluation limitations identified above, we introduce VGGSound-Duet$^{\text{Mask}}$, a pixel-precise version of the VGGSound-Duet benchmark~\cite{mo2023audiovisual}.
We start from the VGGSound-Source (VGG-SS) test set~\cite{Chen2021LVS}, which originally contains 5,158 single-source videos with bounding box annotations.
Due to YouTube's availability, 4,390 videos remain accessible.
For each video, we generate high-quality segmentation masks using SAM~\cite{sam}, with the original bounding boxes as initial prompts.
To construct the multi-source evaluation, we follow the exact pairing protocol of VGGSound-Duet~\cite{mo2023audiovisual}, yielding 3,951 dual-source test pairs by concatenating frames and mixing audio from two videos.
Unlike VGGSound-Instruments, which offers only 446 mask-annotated samples restricted to musical categories, our VGGSound-Duet$^{\text{Mask}}$ benchmark provides nearly 10$\times$ more samples across 220 diverse sound categories.
\cref{tab:benchmark_comparison} summarises existing multi-source localisation benchmarks; our VGGSound-Duet$^{\text{Mask}}$ provides the largest mask-annotated test set.
We detail the construction of this benchmark in \appref{sec:dataset_details}.

\red{A related task that also produces pixel-level masks is audio-visual segmentation (AVS)~\cite{zhou2022avs,zhou2023avss,gao2024avsegformer,guo2025avis}. AVS  targets a different question from ours: they segment the sounding objects in a scene as a whole, either as a single foreground mask in AVSBench~\cite{zhou2022avs}, by semantic category in AVSS~\cite{zhou2023avss}, or as separate instances in AVISeg~\cite{guo2025avis}, and the audio serves only as a cue for whether an object is sounding. Our VGGSound-Duet$^{\text{Mask}}$ instead evaluates the decomposition of a two-source audio mixture into a separate localisation for each source, which no existing AVS benchmark provides.}

\begin{table*}[t]
  \centering
\resizebox{\textwidth}{!}{
  \begin{threeparttable}
  \caption{Summary of Multi-Source Sound Source Localisation Benchmarks.}
  \label{tab:benchmark_comparison}
  \begin{tabular}{llcll}
  \toprule
  Dataset & Year & Test Samples & Annotation Type & \#Classes \\
  \midrule
  MUSIC-Duet~\cite{dsol2020}  & 2018& 17\tnote{a} & Bounding Box & 11  \\
  MUSIC-Synthetic~\cite{dsol2020} & 2020 & 455 & Bounding Box & 15 \\
  VGGSound-Instruments~\cite{mix_and_localize} & 2022 & 446 & Segmentation Mask & 37 \\
  VGGSound-Duet~\cite{mo2023audiovisual} & 2023
    & \textcolor{gray}{\sout{5,158}} 3,951\tnote{b} & Bounding Box & 220 \\
  \textbf{VGGSound-Duet$^{\textbf{Mask}}$ (Ours)} & \textbf{2025} & \textbf{3,951} &
    \textbf{Segmentation Mask} & \textbf{220} \\
  \bottomrule
  \end{tabular}
  \begin{tablenotes}
  \tiny{
    \item[a] MUSIC-Duet reports video-level counts, where each video contains multiple annotated frames, while all other benchmarks count individual frames.
    \item[b] Among the videos released, 3,951 samples are currently (as of Oct. 2025) available; the remaining videos are unavailable on YouTube due to copyright issues or removal by uploaders.
    }
  \end{tablenotes}
  \end{threeparttable}
  }
\end{table*}

\paragraph{\textbf{Quantitative results.}}
\cref{tab:two_mask_datasets_vertical} presents the results under our pixel-precise evaluation protocol.
Transitioning from bounding boxes to segmentation masks changes the evaluation dynamics: models can no longer benefit from activating irrelevant background pixels within rectangular regions.
SCAV's consistent improvement under mask-based evaluation confirms genuine, pixel-level localisation of sound-producing regions rather than reliance on coarse spatial heuristics.

\begin{figure*}[t]
      \centering
      \includegraphics[width=\textwidth]{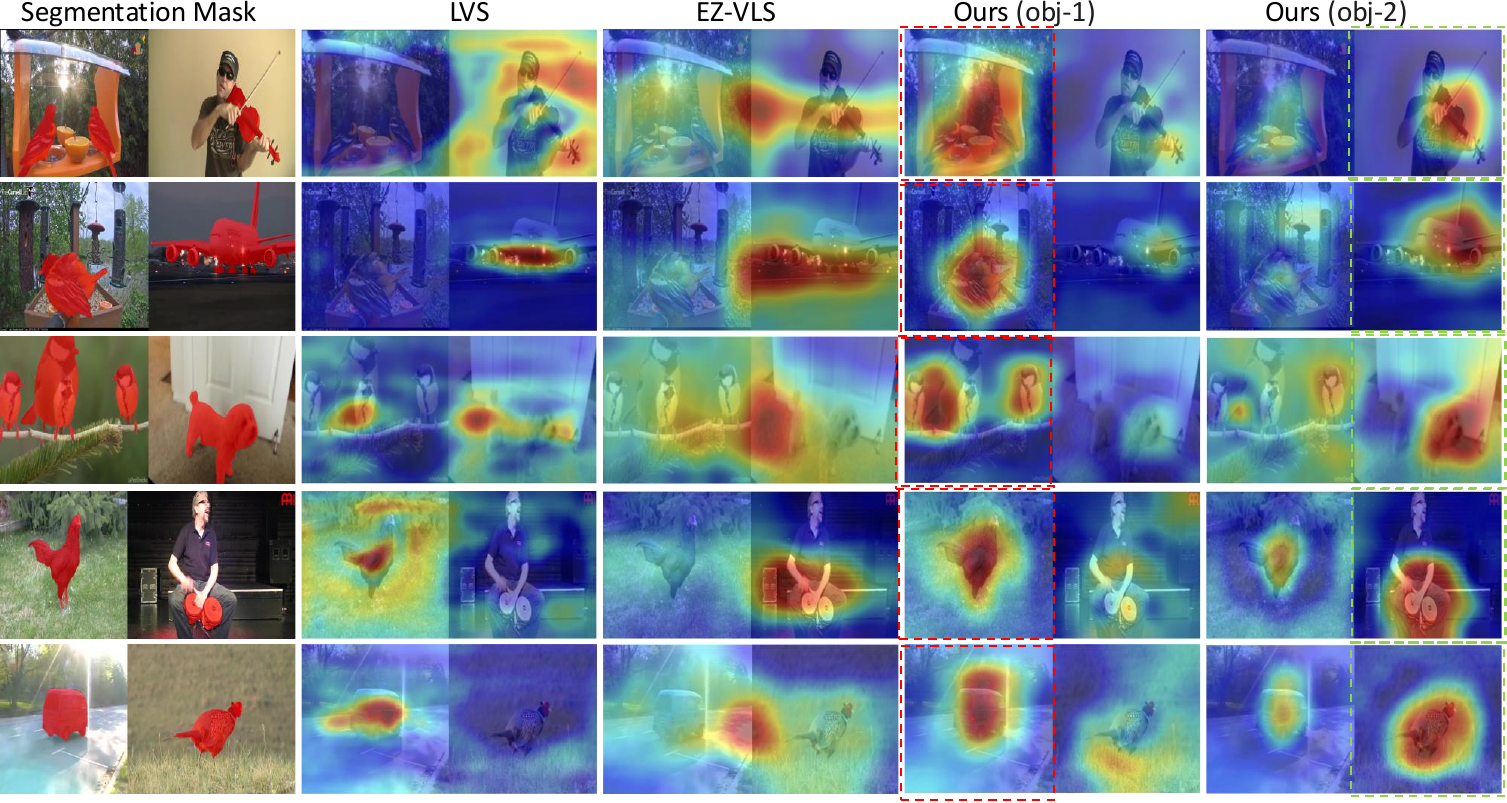}
      \caption{\textbf{Qualitative comparison of multi-source sound localisation.} For each test sample containing mixed audio from two sources, from left to right, we show: \textbf{First}, input frame with segmentation masks indicating both sound sources. \textbf{Second and Third}, localisation heatmaps from baseline models LVS~\cite{Chen2021LVS} and EZ-VSL~\cite{Mo2022EZVSL}. \textbf{Fourth and Last}, localisation heatmaps for each detected source from our SCAV method. The heatmaps visualise predicted sound source locations, where warmer colours indicate higher activation values. Best viewed in colour.}
      \label{fig:qualitative}
\end{figure*}

\paragraph{\textbf{Qualitative results.}} \cref{fig:qualitative} presents the localisation maps generated by different methods on representative multi-source scenarios.
The baseline models LVS~\cite{Chen2021LVS} and EZ-VSL~\cite{Mo2022EZVSL} struggle to localise both sources simultaneously.
As shown in the final two columns, our method generates two independent localisation maps that accurately identify each individual sound source.

Interestingly, in the third row, where multiple birds are present, our method produces connected regions that appropriately cover all birds when localising bird sounds.
This demonstrates its ability to handle spatially distributed sources of the same category.
We also observe that each localisation map exhibits a primary activation pattern: while one source receives strong activation, the other source's region may show weaker residual activations.
This soft separation reflects our progressive decoupling design: our method focuses on each source but maintains contextual awareness rather than enforcing hard boundaries.
This helps when sources have overlapping characteristics or ambiguous boundaries, as the model handles uncertainty naturally.
This qualitative evidence, corroborating our quantitative results, demonstrates that our SCAV method successfully breaks the circular dependency problem and allows each decoupled branch to focus on distinct acoustic sources without interference.
\red{More qualitative comparisons and failure cases are provided in \appref{sec:additional_results}.}

\subsection{Empirical Analysis of Selective Convergence}
To validate selective convergence in the first stage, we conduct an empirical analysis that complements the theoretical justification provided in \appref{sec:supp_theory}.
Specifically, when trained with contrastive learning on mixed audio, the first-stage model naturally focuses on one dominant source rather than attempting to cover all sources equally.
To this end, we design a controlled experiment that directly measures whether the model's predictions concentrate on a single source or spread across multiple sources.

We evaluate the first-stage model on VGGSound-Duet, where each test sample is constructed by concatenating two frames into a $224 \times 448$ image while mixing their corresponding audio waveforms.
For each sample, we have two ground truth masks $\text{GT}_1$ and $\text{GT}_2$, where each mask covers one source region (either the left or right half of the concatenated image) and is zero-padded on the other half. The first-stage model takes the concatenated image and mixed audio as input, and produces a single heatmap $\mathbf{M}_{\text{dom}}$ over the full $224 \times 448$ image domain.
We compute the IoU between the predicted heatmap $\mathbf{M}_{\text{dom}}$ and each ground truth:
\begin{equation}
\text{IoU}_1 = \text{IoU}(\mathbf{M}_{\text{dom}}, \text{GT}_1), \quad \text{IoU}_2 =
\text{IoU}(\mathbf{M}_{\text{dom}}, \text{GT}_2).
\end{equation}

For each sample, we identify the dominant source as the one with a higher IoU:
$i^* = \arg\max(\text{IoU}_1, \text{IoU}_2)$
and denote the other source as $j^*$.
We then compute metrics separately for the dominant source $i^*$ and the second source $j^*$ across all test samples.

\begin{table}[t]
  \centering
  \caption{\textbf{Localisation performance comparison between one dominant source and across dual sources.} It shows that the first-stage model tends to selectively converge on one of the sounding source objects.}
  \label{tab:selective_convergence}
  \resizebox{0.6\columnwidth}{!}{%
  \begin{tabular}{lccc}
    \toprule
    \textbf{Setting} & \textbf{CAP (\%)} & \textbf{CIoU@0.3 (\%)} & \textbf{AUC (\%)} \\
    \midrule
    \textbf{Dominant Source} & 63.39 & 61.70 & 37.82 \\
    \textbf{Second Source} & 13.91 \textcolor{BrickRed}{(-49.48)} & 3.62 \textcolor{BrickRed}{(-58.08)} & 8.12 \textcolor{BrickRed}{(-29.70)} \\
    \bottomrule
  \end{tabular}%
  }
\end{table}

\begin{figure}[t]
      \centering
      \includegraphics[width=\columnwidth]{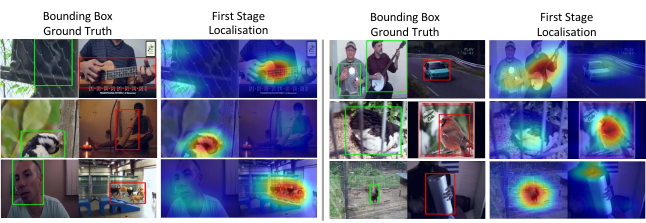}
      \caption{\textbf{Examples of selective convergence behaviour of the first stage.} We show several input sample frames with both ground truth sources (the first and third columns) and the first stage output heatmap (the second and fourth columns). Green boxes indicate the sounding object in the left frame, while red boxes indicate the sounding object in the right frame.
      The heatmaps consistently show strong activation on only one source while producing weak or negligible responses on the other.}
      \label{fig:selective_convergence}
\end{figure}

\cref{tab:selective_convergence} reveals a substantial performance gap between the dominant source and the second source, which provides direct empirical evidence of selective convergence: the model's heatmap predominantly focuses on one source rather than distributing attention across both.
\cref{fig:selective_convergence} further visualises this behaviour across diverse acoustic scenarios.
Both quantitative and qualitative results indicate that the first stage produces focused, dominant-source localisations rather than diffused activations that attempt to cover multiple sources.

\subsection{Ablation Study and Discussion}

\begin{table*}[t]
\centering
\caption{\textbf{Ablation study on VGGSound-Duet.}}
\label{tab:ablation}
\resizebox{\linewidth}{!}{
\begin{tabular}{c|cccc|cccc|cccc}
\toprule
\multirow{2}{*}{\textbf{Setting~\#}}
& \multicolumn{4}{c|}{\textbf{Components}}
& \multicolumn{4}{c|}{\textbf{VGGSound-Duet (BBox-based)}}
& \multicolumn{4}{c}{\textbf{VGGSound-Duet$^{\text{Mask}}$}} \\
\cmidrule(lr){2-5} \cmidrule(lr){6-9} \cmidrule(lr){10-13}
& \footnotesize{Stage~2}
& \footnotesize{Spatial~Prior}
& \footnotesize{Progr.~Train}
& \footnotesize{Cross-neg.}
& \footnotesize{CAP (\%)}
& \footnotesize{CIoU@0.1(\%)}
& \footnotesize{CIoU@0.3(\%)}
& \footnotesize{AUC (\%)}
& \footnotesize{CAP}
& \footnotesize{CIoU@0.1}
& \footnotesize{CIoU@0.3}
& \footnotesize{AUC} \\
\midrule

(a) & \cmark & \xmark$^*$ & \cmark & \xmark
& 37.23 & 73.38 & 28.94 & 21.94
& 27.39 & 52.58 & 16.15 & 15.47 \\

(b) & \xmark & --- & --- & \xmark
& 38.65 & 75.15 & 32.66 & 22.97
& 28.93 & 53.41 & 18.84 & 16.27 \\

(c) & \cmark & \cmark & \xmark & \xmark
& 45.17 & 79.41 & 41.84 & 26.88
& 32.71 & 63.33 & 25.62 & 19.68 \\

(d) & \cmark & \cmark & \cmark & \cmark
& 35.68 & 67.38 & 26.69 & 20.93
& 26.00 & 50.92 & 17.40 & 15.75 \\

\textbf{Ours}
& \cmark & \cmark & \cmark & \xmark
& \textbf{52.96} & \textbf{83.17} & \textbf{53.27} & \textbf{31.96}
& \textbf{39.30} & \textbf{69.42} & \textbf{33.59} & \textbf{23.52} \\

\bottomrule
\multicolumn{13}{l}{\footnotesize $^*$Uses uniform mask instead of learned spatial prior.}
\end{tabular}
}
\end{table*}

We conduct ablation studies on the VGGSound-Duet dataset to analyse the contribution of each component in our SCAV framework, as shown in \cref{tab:ablation}.
The ablation results reveal the interdependencies among components.
Comparing Settings~(a) and~(b) shows that the decoupling architecture relies heavily on spatially informative guidance to function effectively.
Setting~(c) uses the same architecture as our full model but differs only in training strategy: joint end-to-end training versus progressive training.
Despite having identical components, joint training achieves only 41.84\% CIoU@0.3 compared to our 53.27\%.
This demonstrates that the training strategy plays a crucial role in our framework.
In joint training, the first-stage encoder has not yet converged during early training. This leads to unstable spatial priors that mislead the second-stage decoupler.
Progressive training addresses this by first training the first stage to convergence before training the second stage, which yields more stable spatial priors.
\red{Setting~(d) adds cross-sample negatives on top of our full model and degrades performance, since Stage~2 is essentially a within-sample source-assignment problem where the useful contrast lies between the dominant source and the other source inside the same mixture. Such cross-sample negatives instead dilute this within-sample alignment.}

\section{Conclusion}
\label{sec:conclusion}
In this work, we introduced selective convergence for dual-source audio-visual localisation, an emergent property of contrastive audio-visual learning where models naturally focus on dominant sources rather than attempting a joint multi-source representation.
By leveraging selective convergence, we transformed the circular dependency problem inherent in multi-source localisation into a tractable sequence of conditional optimisations.
The proposed SCAV framework achieved \red{the best performance among self-supervised methods in the literature, and} competitive performance compared to weakly-supervised approaches.
Notably, the performance gap over baselines widens on larger-scale datasets, which indicates that our method effectively embraces the richness of large-scale audio-visual data.

Our empirical analysis also revealed systematic biases in current bounding box evaluation protocols, following which we introduced pixel-precise benchmarks that better reflect actual localisation quality.
Future work could explore whether selective convergence generalises to other multi-modal learning scenarios beyond audio-visual correspondence, potentially for handling ambiguous many-to-many associations in a self-supervised manner.

\section*{Acknowledgements}
This project is partially supported by an Amazon Research Award. The computations in this research were partially performed using the Baskerville Tier 2 HPC service. Baskerville was funded by the EPSRC and UKRI through the World Class Labs scheme (EP\textbackslash T022221\textbackslash1) and the Digital Research Infrastructure programme (EP\textbackslash W032244\textbackslash1) and is operated by Advanced Research Computing at the University of Birmingham.

\clearpage

\phantomsection
\addcontentsline{toc}{section}{Appendix}
\renewcommand{\theHsection}{supp.\Alph{section}}
\renewcommand{\theHsubsection}{supp.\Alph{section}.\arabic{subsection}}
\renewcommand{\theHsubsubsection}{supp.\Alph{section}.\arabic{subsection}.\arabic{subsubsection}}
\renewcommand{\theHfigure}{supp.\arabic{figure}}
\renewcommand{\theHtable}{supp.\arabic{table}}
\renewcommand{\theHequation}{supp.\arabic{equation}}
\renewcommand{\theHalgorithm}{supp.\arabic{algorithm}}

\section*{Appendix Roadmap}
\begin{itemize}
    \item In \appref{sec:additional_results}, we present \textbf{additional experimental results}, including extended visual comparisons with multi-source baselines and representative failure cases.
    \item In \appref{sec:dataset_details}, we describe the \textbf{VGGSound-Duet$^{\text{Mask}}$ benchmark}, covering the segmentation mask generation pipeline and handling of sound-producing interactions.
    \item In \appref{sec:supp_algorithm}, we detail the \textbf{dominant source extraction} post-processing algorithm and validate its contribution via ablation.
    \item In \appref{subsec:eval_protocol}, we formalise the \textbf{evaluation protocol}, clarifying the difference between \textit{source-wise} and \textit{frame-wise} evaluation and why \textit{source-wise} evaluation inflates metrics.
    \item In \appref{sec:supp_theory}, we provide a \textbf{theoretical analysis of selective convergence}, showing how the thresholding mechanism amplifies any initial audio-visual correspondence imbalance into a winner-take-all selection of a single source.
    \item In \appref{sec:noprior_reproduction}, we document our \textbf{reproduction of a recent baseline} under a standard frame-wise evaluation protocol, including code modifications for dual-source evaluation, hyperparameter verification, and qualitative comparisons on VGGSound-Duet and MUSIC-Duet.
\end{itemize}
\clearpage

\renewcommand{\thetable}{A\arabic{table}}
\renewcommand{\thefigure}{A\arabic{figure}}
\setcounter{table}{0}
\setcounter{figure}{0}
\setcounter{section}{0}
\renewcommand{\thesection}{\Alph{section}}
\renewcommand{\thesubsection}{\thesection.\arabic{subsection}}

\section{Additional Experimental Results}
\label{sec:additional_results}
\subsection{Extended Visual Comparisons}
The qualitative comparison in \cref{fig:qualitative} covers single-source localisation methods (LVS and EZ-VSL) that produce only a single aggregated heatmap for all sound sources and cannot localise each individual sounding object.
Here, we extend our visual analysis in \cref{fig:extended_comparison} to include Mix-and-Localize~\cite{mix_and_localize}.

\begin{figure*}
      \centering
      \includegraphics[width=\textwidth]{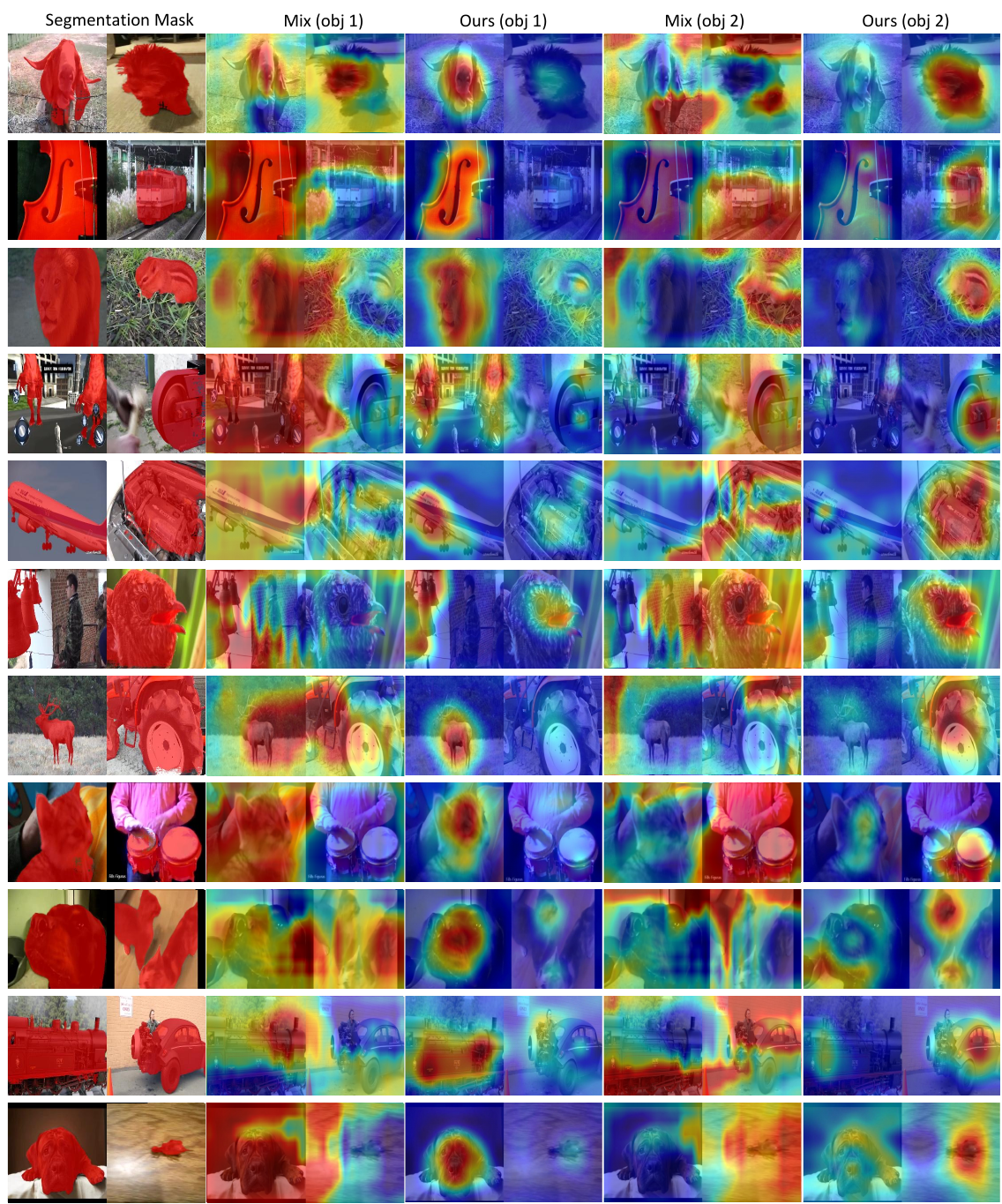}
      \caption{\textbf{Comparison with Mix-and-Localize~\cite{mix_and_localize} on multi-source scenarios.} The first column shows ground truth segmentation masks for both sound sources. The second and third columns present the first source localisation from Mix-and-Localise and our SCAV method, respectively. The fourth and fifth columns display the second source localisation from both methods.}
      \label{fig:extended_comparison}
\end{figure*}

\subsection{Failure Case Analysis}
While our method achieves good performance overall, it fails in several challenging scenarios.
\cref{fig:failure_cases} illustrates three representative failure modes observed in our experiments.

\paragraph{\textbf{Localisation inaccuracy.}} Our method correctly identifies the presence of both sound sources but produces imprecise spatial boundaries. This inaccuracy manifests in two ways: for large objects, the heatmap activations struggle to cover the entire sounding region, while for small objects, the activations extend beyond the actual source boundaries. As shown in the first row's right object, the activation region poorly matches the object's actual size. This limitation might stem from our 7×7 patch resolution, which lacks the fine-grained spatial granularity needed to precisely delineate objects of varying scales.

\begin{figure*}[t]
      \centering
      \includegraphics[width=\textwidth]{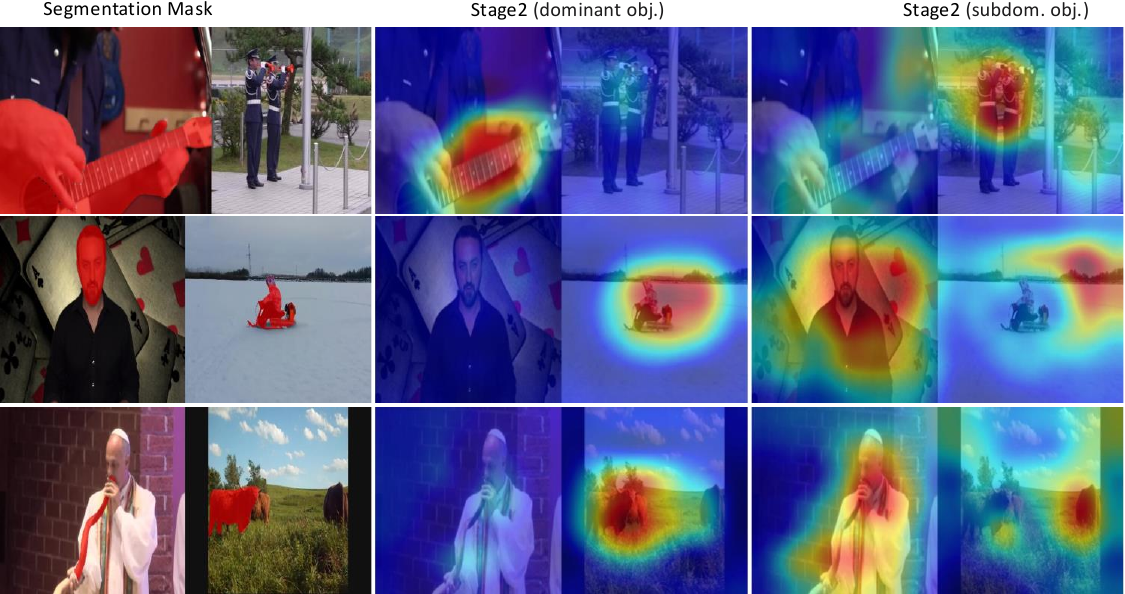}
      \caption{\textbf{Representative failure cases of our method.} The first row demonstrates localisation inaccuracy where activation regions poorly match object boundaries. The second row shows separation failure with spurious activations from multiple sources appearing in a single localisation map. The third row illustrates false activation, where the model activates correlated but non-sounding regions.}
      \label{fig:failure_cases}
\end{figure*}

\paragraph{\textbf{Separation failure.}} One or both localisation maps contain activations from multiple sound sources instead of clean individual separations.
In the second row's right example in \cref{fig:failure_cases}, the human speech localisation map shows spurious activations near the motorcycle.
One possible explanation is that audio source separation becomes more challenging when the source volumes are highly imbalanced. In this case, the motorcycle's engine noise is louder in the mixed audio, which may hinder the extraction of the human speech component.
As a result, residual motorcycle sounds in the separated human audio channel could lead to false activations near the motorcycle during audio-visual similarity computation.

\paragraph{\textbf{False activation.}} The model activates incorrect visual regions that correlate with but do not produce the actual sound. In the left example of the third row, where a man plays a wind instrument, our ground truth correctly annotates the instrument as the sound source, whereas the model's activation focuses on the man's mouth region.
This mislocalisation may be explained by several factors. First, our motion-based approach may favour regions with stronger visual dynamics, and the mouth movements are more pronounced than the subtle motion of the instrument. Second, the instrument's thin profile and light colouration may make it less visually distinctive from the background, which could encourage the model to focus on the more visually salient facial features despite their indirect relationship to sound production.

These failure modes suggest several potential limitations of self-supervised audio-visual learning, including the challenge of scale-aware localisation with fixed patch sizes, the difficulty of separating severely imbalanced audio sources, and a tendency to focus on visually salient features rather than the actual sound-producing objects. Future work may address these issues through multi-scale architectures, improved audio separation techniques, and stronger semantic understanding of sound production mechanisms.

\section{Dataset Details}
\label{sec:dataset_details}
We introduce VGGSound-Duet$^{\text{Mask}}$, a pixel-precise evaluation benchmark for multi-source sound localisation.
Our VGGSound-Duet$^{\text{Mask}}$ provides the largest test set with pixel-level segmentation masks, covering 220 sound classes across 3,951 samples.

\subsection{Segmentation Mask Generation}
Our annotation pipeline leverages the Segment Anything Model (SAM)~\cite{sam} to convert bounding boxes into precise segmentation masks through a multi-stage process:

First, for each test sample, we use the bounding box provided by VGGSound-Source~\cite{Chen2021LVS} as a prompt for SAM to generate initial segmentation proposals. SAM's box-prompted segmentation produces three candidate masks with associated confidence scores.

Next, annotators review the three automatically generated mask candidates alongside the original image and bounding box. The interface displays all candidates simultaneously for comparison.
Annotators then select the most accurate segmentation from these options.

Then, when none of the box-prompted candidates adequately capture the sound source, annotators generate new masks with point prompts. They click on specific locations within the object to guide SAM toward more accurate segmentation boundaries.

Finally, the selected masks are saved in pickle format as binary arrays where sounding source regions have a pixel value of ``1'' and ``0'' otherwise.

\subsection{Handling Sound-Producing Interactions}
A critical challenge in sound source annotation arises when sounds result from interactions between
multiple objects.
We establish clear guidelines to ensure consistency across different interaction scenarios.

\paragraph{\textbf{Musical instruments.}} For musical instruments played by humans, we follow the annotation protocol established by VGGSound-Instruments~\cite{mix_and_localize}.
We annotate only the instrument itself, not the performer or their hands.
This decision maintains semantic consistency: the instrument is the primary sound source, while human interaction serves only as the activation mechanism.

\paragraph{\textbf{Object interactions.}}
For sounds produced through interactions between non-musical objects, we annotate objects within the bounding box provided by VGGSound-Source that actively participate in sound production.
For example, in ``sawing trees'', the VGGSound-Source bounding box encompasses both the saw and the portion of the tree being cut.
Following this guidance, our segmentation mask includes both the saw and the specific tree region within the bounding box.

\section{Post-processing for Dominant Source Extraction}
\label{sec:supp_algorithm}

\begin{algorithm}
\caption{Post-processing for Dominant Source Extraction}
\label{alg:dominant_extraction}
\definecolor{codeblue}{rgb}{0.25,0.5,0.5}
\begin{algorithmic}
  \STATE \textbf{Input:} Similarity map $\mathcal{S} \in \mathbb{R}^{h \times w}$,
         Visual features $V' \in \mathbb{R}^{D \times h \times w}$
  \STATE \textbf{Output:} Dominant source map $\mathbf{M}_{\text{dom}} \in [0,1]^{h \times w}$
  \STATE
  \STATE \textcolor{codeblue}{// Step 1: Extract background embedding}
  \STATE $N(i,j) = 1 - \sigma\bigl((\mathcal{S}(i,j) - \epsilon_p)/\tau\bigr)$
         \hfill $\triangleright$ soft background mask, $\in \mathbb{R}^{h \times w}$
  \STATE $v_{\text{neg}} = \frac{\sum_{i,j} N(i,j) \cdot \mathcal{S}(i,j) \cdot V'(:,i,j)}
         {\sum_{i,j} N(i,j) \cdot \mathcal{S}(i,j)}$
         \hfill $\triangleright$ weighted mean, $\in \mathbb{R}^{D}$
  \STATE
  \STATE \textcolor{codeblue}{// Step 2: Extract dominant source seed}
  \STATE $(i^*, j^*) = \arg\max_{i,j}\, \mathcal{S}(i,j)$
  \STATE $v_1 = V'(:, i^*, j^*)$
         \hfill $\triangleright$ seed embedding, $\in \mathbb{R}^{D}$
  \STATE
  \STATE \textcolor{codeblue}{// Step 3: Per-patch distance comparison}
  \FOR{each position $(i,j)$}
    \STATE $d^+(i,j) = \| V'(:,i,j) - v_1 \|_2$
           \hfill $\triangleright$ distance to seed
    \STATE $d^-(i,j) = \| V'(:,i,j) - v_{\text{neg}} \|_2$
           \hfill $\triangleright$ distance to background
    \STATE $\mathbf{M}_{\text{dom}}(i,j) = d^-(i,j) - d^+(i,j)$
  \ENDFOR
  \STATE
  \STATE \textcolor{codeblue}{// Step 4: Normalize to $[0,1]$}
  \STATE $\mathbf{M}_{\text{dom}} \leftarrow
         \frac{\mathbf{M}_{\text{dom}} - \min(\mathbf{M}_{\text{dom}})}
              {\max(\mathbf{M}_{\text{dom}}) - \min(\mathbf{M}_{\text{dom}})}$
  \STATE
  \RETURN $\mathbf{M}_{\text{dom}}$
  \end{algorithmic}
\end{algorithm}

As mentioned in \cref{subsec:stage1}, after the first stage produces a similarity map $\mathcal{S}$ via selective convergence, we apply \cref{alg:dominant_extraction} to extract a cleaner dominant source localisation map $\mathbf{M}_{\text{dom}}$.
The audio-visual similarity map obtained by Stage1 under selective convergence may contain weak noise or residual activations.

\begin{figure*}[t]
    \centering
    \includegraphics[width=0.95\textwidth]{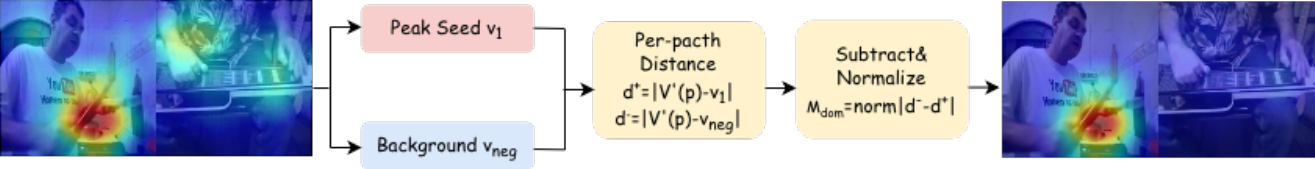}
    \caption{\textbf{Illustration of the post-processing algorithm for dominant source extraction.}
    The pipeline takes the raw similarity map $\mathcal{S}$ (left) with residual activations and produces a clean dominant source map   $\mathbf{M}_{\text{dom}}$ (right).
    The processing flow extracts two reference embeddings: the peak seed $v_1$ from the strongest activation and the background embedding $v_{\text{neg}}$ from low-activation regions, then assigns each spatial location based on its relative feature distance to the seed versus the background.}
    \label{fig:supp_algorithm}
\end{figure*}

\cref{fig:supp_algorithm} illustrates this process: \cref{alg:dominant_extraction} refines this map using visual feature similarity: for each location, it computes whether the visual feature is more similar to the peak response or to the background, and assigns response values accordingly.

\begin{table*}[h!]
\centering
\caption{\textbf{Ablation study on VGGSound-Duet evaluating the contribution of \cref{alg:dominant_extraction}.}}
\label{tab:ablation_algorithm1}
\resizebox{\linewidth}{!}{
\begin{tabular}{c|cccc|cccc}
\toprule
\multirow{2}{*}{\textbf{Components}}
& \multicolumn{4}{c|}{\textbf{VGGSound-Duet (BBox-based)}}
& \multicolumn{4}{c}{\textbf{VGGSound-Duet$^{\text{Mask}}$}} \\
\cmidrule(lr){2-5} \cmidrule(lr){6-9}
& \footnotesize{CAP (\%)}
& \footnotesize{CIoU@0.1(\%)}
& \footnotesize{CIoU@0.3(\%)}
& \footnotesize{AUC (\%)}
& \footnotesize{CAP}
& \footnotesize{CIoU@0.1}
& \footnotesize{CIoU@0.3}
& \footnotesize{AUC} \\
\midrule

\footnotesize{w/o \cref{alg:dominant_extraction}}
& 52.32 & 83.06 & 52.44 & 31.67
& 38.93 & 69.38 & 33.26 & 23.40 \\

\footnotesize{w/ \cref{alg:dominant_extraction}}
& \textbf{52.96} & \textbf{83.17} & \textbf{53.27} & \textbf{31.96}
& \textbf{39.30} & \textbf{69.42} & \textbf{33.59} & \textbf{23.52} \\

\bottomrule
\end{tabular}
}
\end{table*}

As shown in \cref{tab:ablation_algorithm1}, removing \cref{alg:dominant_extraction} leads to only a marginal performance drop. This suggests that selective convergence alone is sufficient to provide effective spatial priors, which allows Stage 2 to remain performant even without the additional refinement.

\section{Evaluation Protocol}
\label{subsec:eval_protocol}

Early multi-source sound localisation works, including Mix-and-Localize~\cite{mix_and_localize} and DSOL~\cite{dsol2020}, evaluate on MUSIC-Duet~\cite{sound_of_pixel} with a \emph{frame-wise} protocol: each predicted heatmap is compared against the ground-truth mask over the \emph{full} image frame.
This is a natural choice because MUSIC-Duet consists of real duet recordings in which two instruments perform simultaneously within the same video; the two sound sources may appear at arbitrary spatial locations, and there is no synthetic spatial partition that would allow the evaluation to be restricted to a sub-region of the frame.

Subsequently, Mo~\etal~\cite{mo2023audiovisual} proposed VGGSound-Duet, a \emph{synthetic} dual-source benchmark in which each test sample is constructed by horizontal concatenation of two single-source frames into one $H{\times}2W$ image ($H{=}W{=}224$) (see \cref{fig:eval_protocol}, left), and their audio waveforms are summed to form the mixed audio signal.
In their released evaluation code, Mo~\etal adopted a \emph{source-wise} protocol: each predicted heatmap is cropped to the half-frame that contains its designated source, and IoU is computed within that $H{\times}W$ region only. Kim~\etal~\cite{kim2024learning} subsequently followed the same protocol.

In this work, the results reported in \cref{tab:two_bbox_datasets,tab:two_mask_datasets_vertical} use the frame-wise protocol, consistent with Mix-and-Localize~\cite{mix_and_localize} and DSOL~\cite{dsol2020}. This also ensures a unified evaluation standard across MUSIC-Duet~\cite{sound_of_pixel} and VGGSound-Duet~\cite{mo2023audiovisual}, since MUSIC-Duet naturally requires frame-wise evaluation and the same protocol on VGGSound-Duet allows results to be directly compared across the two benchmarks. Below, we formally define both protocols, present their mathematical relationship, and illustrate them in \cref{fig:eval_protocol}.

\paragraph{\textbf{Notations.}}
Let $(i,j)$ denote pixel coordinates with $i \in \{0,\ldots,H{-}1\}$ for rows and $j \in \{0,\ldots,2W{-}1\}$ for columns. By construction, Source~1 is entirely contained within the left half $\Omega_L = \{(i,j) \mid j < W\}$ and Source~2 within the right half $\Omega_R = \{(i,j) \mid j \ge W\}$. Note that this left--right separation is purely an artefact of the synthetic construction; in real-world scenes, multiple sound sources may appear at arbitrary spatial locations.
The model therefore cannot exploit this spatial layout as a prior and must localise each source based solely on audio-visual correspondence.
The ground truth comprises two independent half-frame masks:
$\mathcal{M}_1^{\mathrm{half}} \in \{0,1\}^{H\times W}$ for the sounding object in $\Omega_L$, and
$\mathcal{M}_2^{\mathrm{half}} \in \{0,1\}^{H\times W}$ for the sounding object in $\Omega_R$.

\paragraph{\textbf{Frame-wise evaluation.}}
Under frame-wise evaluation, each half-frame ground-truth mask is zero-padded to the full resolution:
\begin{align}
  \mathcal{M}_1^{+} &= \bigl[\mathcal{M}_1^{\mathrm{half}},\;
                        \mathbf{0}^{H\times W}\bigr]
                      \in \{0,1\}^{H\times 2W}, \label{eq:gt_pad1}\\
  \mathcal{M}_2^{+} &= \bigl[\mathbf{0}^{H\times W},\;
                        \mathcal{M}_2^{\mathrm{half}}\bigr]
                      \in \{0,1\}^{H\times 2W}. \label{eq:gt_pad2}
\end{align}
Each heatmap $\mathcal{H}_k \in [0,1]^{H\times 2W}$ is binarised to obtain a full-frame binary mask
$\mathcal{B}_k \in \{0,1\}^{H\times 2W}$.
IoU is then computed over the full domain $\Omega = \Omega_L \cup \Omega_R$:
\begin{equation}
  \mathrm{IoU}^{\,\text{frame}}_{k}
  = \frac{|\,\mathcal{B}_k \;\cap\; \mathcal{M}_k^{+}|}
         {|\,\mathcal{B}_k \;\cup\; \mathcal{M}_k^{+}|},
  \qquad k \in \{1,2\},
  \label{eq:iou_def}
\end{equation}
where both $\mathcal{B}_k$ and $\mathcal{M}_k^{+}$ span the entire $H{\times}2W$ image.

\paragraph{\textbf{Source-wise evaluation.}}
Under source-wise evaluation, each predicted heatmap is \emph{cropped} to the half-frame that contains its designated source, and IoU is computed entirely within that $H{\times}W$ region (see \cref{fig:eval_protocol}, middle).
Concretely, the model produces two continuous heatmaps $\mathcal{H}_1, \mathcal{H}_2 \in [0,1]^{H\times 2W}$ over the full concatenated image.
Each heatmap is then binarised and cropped to the corresponding half: $\hat{\mathcal{B}}_1 \in \{0,1\}^{H\times W}$ covers only $\Omega_L$, and $\hat{\mathcal{B}}_2 \in \{0,1\}^{H\times W}$ covers only $\Omega_R$.
IoU is computed against the half-frame ground-truth mask:
\begin{equation}
  \mathrm{IoU}^{\,\text{src}}_{k}
  = \frac{|\,\hat{\mathcal{B}}_k \;\cap\; \mathcal{M}_k^{\mathrm{half}}|}
         {|\,\hat{\mathcal{B}}_k \;\cup\; \mathcal{M}_k^{\mathrm{half}}|},
  \qquad k \in \{1,2\},
  \label{eq:iou_source}
\end{equation}
where both $\hat{\mathcal{B}}_k$ and $\mathcal{M}_k^{\mathrm{half}}$ are defined over the same $H{\times}W$ domain.
Any activations that the model produces on the \emph{opposite} half are discarded before the metric is evaluated.

We verify this protocol from the official evaluation code released with Mo~\etal~\cite{mo2023audiovisual}.
The key steps of the \texttt{validate\_multi} function are reproduced below:

\begin{lstlisting}
for j in range(num_mixtures):
    avl_map = model(image[:,j].float(),
                    spec.float(), ...)[1].unsqueeze(1)
    avl_map = F.interpolate(avl_map, size=(224, 224), ...)
    avl_map_list.append(avl_map)

gt_map = bboxes['gt_map'][i].numpy()
for j in range(num_mixtures):
    pred = normalize_img(avl_map[i, j, 0])
    if j == 0:
        evaluator_0.update(bb, gt_map[j], conf, pred, thr, name[i])
    elif j == 1:
        evaluator_1.update(bb, gt_map[j], conf, pred, thr, name[i])
\end{lstlisting}

\begin{lstlisting}
def update(self, bb, gt, conf, pred, pred_thr, name):
    infer = np.zeros((224, 224))
    infer[pred >= pred_thr] = 1
    ciou = np.sum(infer * gt) / (np.sum(gt)
           + np.sum(infer * (gt == 0)))
\end{lstlisting}

Three observations confirm the source-wise protocol:
(1)~The model performs a separate forward pass for each single-source frame \texttt{image[:,j]} of size  $H{\times}W$, producing a $224{\times}224$ heatmap per source (\texttt{num\_mixtures\,=\,2}).
The two sources are never jointly evaluated on the concatenated $H{\times}2W$ image.
(2)~Each heatmap is compared only against its corresponding half-frame ground truth \texttt{gt\_map[j]} of shape $224{\times}224$;
Activations that the model might produce for the \emph{other} source are never visible to the evaluator.
(3)~Inside \texttt{EvaluatorFull.update()}, the binary prediction mask \texttt{infer} is explicitly allocated as a $224{\times}224$ array. This indicates that IoU is computed entirely within the $H{\times}W$ half-frame domain. As a result, cross-source false positives are structurally excluded from the IoU computation.

\paragraph{\textbf{Hungarian matching.}}
Under both protocols, the two predicted heatmaps have no inherent ordering, so we apply Hungarian matching to find the optimal assignment before computing any metric.
For the $n$-th test sample with predictions $\mathcal{H}_{n,1}, \mathcal{H}_{n,2} \in [0,1]^{H\times 2W}$ and ground-truth masks $\mathcal{M}_{n,1}, \mathcal{M}_{n,2}$, we select the permutation
\begin{equation}
  \sigma_n^{*} = \argmax_{\sigma \in \mathfrak{S}_2}
  \sum_{k=1}^{2}
    \mathrm{IoU}\!\bigl(\mathcal{H}_{n,k},\;\mathcal{M}_{n,\sigma(k)}\bigr),
  \label{eq:hungarian}
\end{equation}
and re-index the ground truth accordingly.
After matching, each sample contributes two source-level (prediction, GT) pairs.
The metrics CAP, CIoU, and AUC defined in \cref{subsec:metrics} are then computed over all $K{=}2N$ such pairs (where $N$ is the number of test samples), with each pair treated as an independent class $k$ with $\delta_k{=}1$.

\paragraph{\textbf{Mathematical relationship between the two protocols.}}
The source-wise and frame-wise formulations share the same numerator: because $\mathcal{M}_k^{+}$ is zero on the opposite half, the intersection
$|\mathcal{B}_k \cap \mathcal{M}_k^{+}|$ counts only pixels in the \emph{target} half, which is identical to $|\hat{\mathcal{B}}_k \cap \mathcal{M}_k^{\mathrm{half}}|$.
The difference lies entirely in the denominator.
Let $f_{\bar{k}} = \sum_{p\in\Omega_{\bar{k}}} \mathcal{B}_k(p)$ denote the number of activations on the \emph{opposite} half (\ie cross-source activations).
Then the two denominators relate as:
\begin{equation}
  \underbrace{|\,\mathcal{B}_k \;\cup\; \mathcal{M}_k^{+}\,|}_{\text{frame-wise (Eq.~\eqref{eq:iou_def})}}
  \;=\;
  \underbrace{|\,\hat{\mathcal{B}}_k \;\cup\; \mathcal{M}_k^{\mathrm{half}}\,|}_{\text{source-wise (Eq.~\eqref{eq:iou_source})}}
  \;+\; f_{\bar{k}}.
  \label{eq:denom_compare}
\end{equation}
Since $f_{\bar{k}} \ge 0$ and the numerators are identical, it follows that
\begin{equation}
  \mathrm{IoU}^{\,\text{src}}_{k} \;\ge\; \mathrm{IoU}^{\,\text{frame}}_{k},
  \label{eq:iou_inequality}
\end{equation}
with equality if and only if $f_{\bar{k}} = 0$, \ie the model produces zero activation on the opposite half.
The two protocols therefore yield identical results when a model activates exclusively within its target half, and differ only when cross-source activations are present.

\paragraph{\textbf{What each protocol measures.}}
The two protocols correspond to different aspects of multi-source localisation.
Source-wise evaluation measures \emph{detection}: whether the model can identify the correct region within the half-frame that contains each source.
Because IoU is computed within the $H{\times}W$ half-frame only, any activations on the opposite half are not considered (see \cref{fig:eval_protocol}, middle).
Frame-wise evaluation measures both \emph{detection} and \emph{separation}: the model must not only activate the correct region for each source, but also suppress activations at the location of the other source.
As shown in Eq.~\eqref{eq:iou_inequality}, when a model produces activations on the opposite half ($f_{\bar{k}} > 0$), these contribute to the denominator under frame-wise evaluation but not under source-wise evaluation, causing the two protocols to yield different scores for the same prediction.

\begin{figure}[t]
  \centering
  \includegraphics[width=\columnwidth]{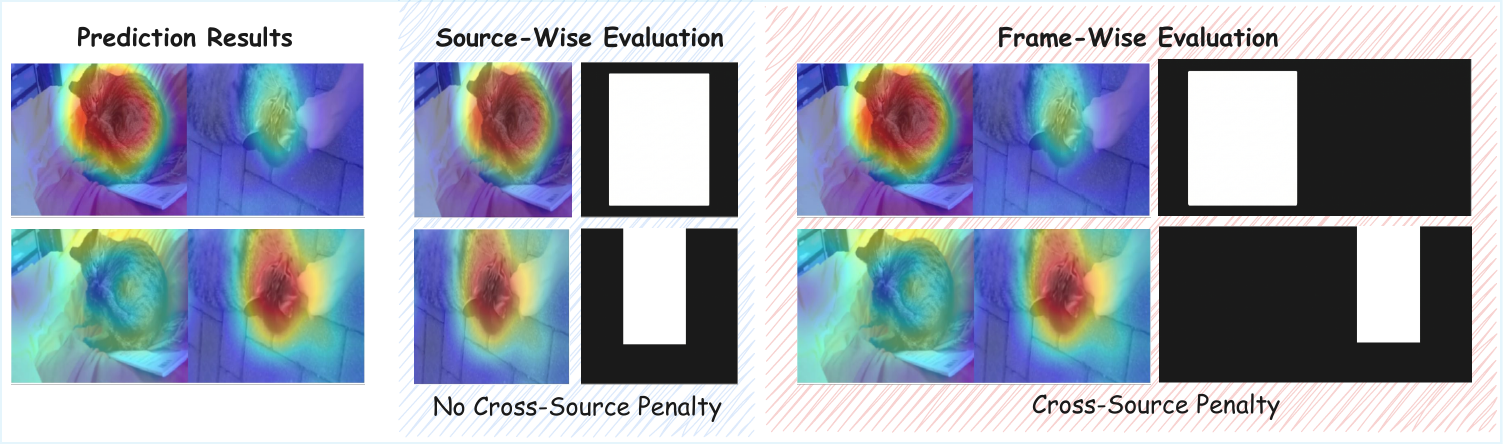}
  \caption{\textbf{Source-wise vs.\ frame-wise evaluation protocol.} \textit{Left}: A synthetic dual-source test sample is constructed by horizontally concatenating two single-source frames. \textit{Middle}: Source-wise evaluation restricts each heatmap to the corresponding half-frame, so cross-source activations are never penalised, which inflates all IoU-based metrics. \textit{Right}: Frame-wise evaluation assesses each heatmap over the full frame against a zero-padded ground-truth mask, penalising any failure to suppress the non-target side.}
  \label{fig:eval_protocol}
\end{figure}

\section{Theoretical Analysis of Selective Convergence}
\label{sec:supp_theory}

We provide a formal analysis of why Stage~1 exhibits selective convergence. The goal is to explain why the model commits to a \emph{single} source rather than represent both equally; it is \emph{not} to predict which particular source is selected. Consistent with the main text, we define the dominant source as whichever source the model converges to, and our framework depends only on the fact that it converges to one of them, not on which one. Our analysis, which adapts the simplicity-bias framework of Xue~\etal~\cite{xue2023icml}, shows that any initial imbalance in audio-visual correspondence is amplified by the differentiable thresholding mechanism into a winner-take-all dynamic.

\subsection{Problem Setup and Assumptions}

\paragraph{Data model.}
Consider a training sample with visual frame $I$ containing two sound sources in disjoint spatial regions $\Omega_1, \Omega_2 \subset \{1,\ldots,h\}\times\{1,\ldots,w\}$.
The mixed audio is encoded as $\mathbf{a}_{\text{mix}} = \mathbf{a}_1 + \mathbf{a}_2 \in \mathbb{R}^D$, where $\mathbf{a}_k$ denotes the contribution of source~$k$ to the mixture.
The visual encoder produces spatial features $V'(p) \in \mathbb{R}^D$ at each position~$p$.

To enable rigorous analysis, we make the following assumptions, which parallel the feature decomposition framework of Xue~\etal~\cite{xue2023icml}:

\begin{assumption}[Feature Orthogonality]
\label{assump:orthogonal}
There exist two orthonormal basis vectors $\mathbf{u}_1, \mathbf{u}_2 \in \mathbb{R}^D$ such that:
\begin{equation}
  \mathbf{a}_1 = \alpha_1 \mathbf{u}_1, \quad \mathbf{a}_2 = \alpha_2 \mathbf{u}_2, \quad \langle \mathbf{u}_1, \mathbf{u}_2 \rangle = 0,
\end{equation}
where $\alpha_1, \alpha_2 > 0$ denote the audio response magnitudes of the two sources along their respective directions.
\end{assumption}

\begin{assumption}[Visual Feature Locality]
\label{assump:locality}
For position $p \in \Omega_k$, the visual feature decomposes as:
\begin{equation}
  V'(p) = \beta_{k,p} \mathbf{u}_k + \boldsymbol{\epsilon}_p,
\end{equation}
where $\beta_{k,p} > 0$ is the matching strength, and $\boldsymbol{\epsilon}_p$ is noise satisfying $\mathbb{E}[\boldsymbol{\epsilon}_p] = \mathbf{0}$ and $\|\boldsymbol{\epsilon}_p\| \leq \sigma$.
\end{assumption}

\begin{assumption}[Low Noise]
\label{assump:noise}
The noise level is small relative to the signal:
\begin{equation}
  \sigma \ll \min(\alpha_1 \bar{\beta}_1, \alpha_2 \bar{\beta}_2).
\end{equation}
\end{assumption}

For brevity we write $g_k := \alpha_k \bar{\beta}_k$ for the average similarity that region $\Omega_k$ attains at initialisation, where $\bar{\beta}_k := \mathbb{E}_{p \in \Omega_k}[\beta_{k,p}]$ is the average visual matching strength in its region. Generically, the two regions are not exactly balanced; without loss of generality, we label the leading region --- the one the model ends up converging to --- as source~1, so that $g_1 > g_2$. This labels the eventual winner rather than predicting it: our analysis characterises how an existing gap is amplified, not what creates it.

\subsection{Similarity Decomposition and Regional Comparison}

\begin{lemma}[Similarity Decomposition]
\label{lem:decomp}
Under Assumptions~\ref{assump:orthogonal} and~\ref{assump:locality}, the similarity map decomposes as:
\begin{equation}
  \mathcal{S}(p) = \langle V'(p), \mathbf{a}_{\text{mix}} \rangle =
  \begin{cases}
    \alpha_1 \beta_{1,p} + \xi_p, & p \in \Omega_1 \\
    \alpha_2 \beta_{2,p} + \xi_p, & p \in \Omega_2
  \end{cases}
\end{equation}
where $\xi_p = \langle \boldsymbol{\epsilon}_p, \alpha_1 \mathbf{u}_1 + \alpha_2 \mathbf{u}_2 \rangle$ is a noise term with $|\xi_p| \leq \sigma(\alpha_1 + \alpha_2)$.
\end{lemma}

\begin{proof}
For $p \in \Omega_1$, substituting the decompositions from Assumptions~\ref{assump:orthogonal} and~\ref{assump:locality}:
\begin{align}
  \mathcal{S}(p) &= \langle \beta_{1,p} \mathbf{u}_1 + \boldsymbol{\epsilon}_p, \alpha_1 \mathbf{u}_1 + \alpha_2 \mathbf{u}_2 \rangle \\
  &= \beta_{1,p} \alpha_1 \langle \mathbf{u}_1, \mathbf{u}_1 \rangle + \beta_{1,p} \alpha_2 \langle \mathbf{u}_1, \mathbf{u}_2 \rangle + \langle \boldsymbol{\epsilon}_p, \alpha_1 \mathbf{u}_1 + \alpha_2 \mathbf{u}_2 \rangle \\
  &= \alpha_1 \beta_{1,p} + \xi_p,
\end{align}
where we used $\langle \mathbf{u}_1, \mathbf{u}_1 \rangle = 1$ and $\langle \mathbf{u}_1, \mathbf{u}_2 \rangle = 0$.
The bound on $|\xi_p|$ follows from Cauchy-Schwarz: $|\xi_p| \leq \|\boldsymbol{\epsilon}_p\| \cdot \|\alpha_1 \mathbf{u}_1 + \alpha_2 \mathbf{u}_2\| \leq \sigma(\alpha_1 + \alpha_2)$.
The case $p \in \Omega_2$ is analogous.
\end{proof}

\begin{proposition}[Regional Similarity Gap]
\label{prop:gap}
Under Assumptions~\ref{assump:orthogonal}--\ref{assump:noise}, the average similarity in region~$\Omega_1$ exceeds that in region~$\Omega_2$:
\begin{equation}
  \bar{\mathcal{S}}_1 := \mathbb{E}_{p \in \Omega_1}[\mathcal{S}(p)] > \bar{\mathcal{S}}_2 := \mathbb{E}_{p \in \Omega_2}[\mathcal{S}(p)].
\end{equation}
Specifically, $\bar{\mathcal{S}}_1 \approx \alpha_1 \bar{\beta}_1 = g_1$ and $\bar{\mathcal{S}}_2 \approx \alpha_2 \bar{\beta}_2 = g_2$, so the gap $\bar{\mathcal{S}}_1 - \bar{\mathcal{S}}_2 \approx g_1 - g_2 > 0$ follows directly from the labeling $g_1 > g_2$.
\end{proposition}

\begin{proof}
Taking expectations in Lemma~\ref{lem:decomp}:
\begin{equation}
  \bar{\mathcal{S}}_1 = \alpha_1 \bar{\beta}_1 + \mathbb{E}[\xi_p], \quad \bar{\mathcal{S}}_2 = \alpha_2 \bar{\beta}_2 + \mathbb{E}[\xi_p].
\end{equation}
Since $\mathbb{E}[\boldsymbol{\epsilon}_p] = \mathbf{0}$, we have $\mathbb{E}[\xi_p] = 0$.
Hence $\bar{\mathcal{S}}_1 \approx g_1 > g_2 \approx \bar{\mathcal{S}}_2$ by the definition of the dominant source.
\end{proof}

\subsection{Thresholding-Induced Mask Concentration}

The positive score in the contrastive loss (\cref{eq:positive_score}) is:
\begin{equation}
  P = \frac{1}{|\hat{m}_p|} \sum_{p} \hat{m}_p(p) \cdot \mathcal{S}(p),
\end{equation}
where the positive mask $\hat{m}_p(p) = \sigma\!\bigl((\mathcal{S}(p) - \epsilon_p)/\tau\bigr)$ with small temperature~$\tau$ approximates a step function.

\begin{lemma}[Mask Concentration]
\label{lem:mask}
Suppose the threshold satisfies $\bar{\mathcal{S}}_1 - \delta > \epsilon_p > \bar{\mathcal{S}}_2 + \delta$ for some $\delta > 0$.
Under Assumption~\ref{assump:noise}, with probability at least $1 - \exp(-\Omega(\delta^2/\sigma^2))$, the mask concentrates on region~$\Omega_1$:
\begin{equation}
  \sum_{p \in \Omega_1} \hat{m}_p(p) \geq (1-\epsilon)|\Omega_1|, \quad \sum_{p \in \Omega_2} \hat{m}_p(p) \leq \epsilon|\Omega_2|,
\end{equation}
for arbitrarily small $\epsilon > 0$ as $\tau \to 0$.
\end{lemma}

\begin{proof}
By Lemma~\ref{lem:decomp}, for $p \in \Omega_1$, $\mathcal{S}(p) = \alpha_1 \beta_{1,p} + \xi_p$.
Since $|\xi_p| \leq \sigma(\alpha_1 + \alpha_2)$ and $\sigma$ is small (Assumption~\ref{assump:noise}), with high probability $\mathcal{S}(p) \approx \alpha_1 \beta_{1,p} \approx \bar{\mathcal{S}}_1 > \epsilon_p$.
As $\tau \to 0$, $\sigma((\mathcal{S}(p) - \epsilon_p)/\tau) \to 1$ for such positions.
Similarly, for $p \in \Omega_2$, $\mathcal{S}(p) \approx \bar{\mathcal{S}}_2 < \epsilon_p$, so $\hat{m}_p(p) \to 0$.
The probability bound follows from standard concentration inequalities for bounded random variables.
\end{proof}

\subsection{Gradient Flow and Feature Suppression}

\begin{theorem}[Selective Convergence via Feature Suppression]
\label{thm:main}
Under Assumptions~\ref{assump:orthogonal}--\ref{assump:noise}, suppose the audio representation at step~$t$ is $\mathbf{a}^{(t)} = c_1^{(t)} \mathbf{u}_1 + c_2^{(t)} \mathbf{u}_2$.
If the threshold satisfies the condition in Lemma~\ref{lem:mask}, then the gradient update predominantly increases $c_1$ while leaving $c_2$ unchanged:
\begin{equation}
  c_1^{(t+1)} = c_1^{(t)} + \eta \bar{\beta}_1 + O(\epsilon), \quad c_2^{(t+1)} = c_2^{(t)} + O(\epsilon),
\end{equation}
where $\eta$ is the learning rate and $\epsilon \to 0$ as $\tau \to 0$.
\end{theorem}

\begin{proof}
The gradient of $P$ with respect to $\mathbf{a}$ is:
\begin{equation}
  \frac{\partial P}{\partial \mathbf{a}} = \frac{1}{|\hat{m}_p|} \sum_{p} \hat{m}_p(p) \cdot V'(p).
\end{equation}

By Lemma~\ref{lem:mask}, the sum is dominated by positions in $\Omega_1$:
\begin{equation}
  \frac{\partial P}{\partial \mathbf{a}} \approx \frac{1}{|\Omega_1|} \sum_{p \in \Omega_1} V'(p).
\end{equation}

Substituting Assumption~\ref{assump:locality} ($V'(p) = \beta_{1,p} \mathbf{u}_1 + \boldsymbol{\epsilon}_p$ for $p \in \Omega_1$):
\begin{equation}
  \frac{\partial P}{\partial \mathbf{a}} \approx \frac{1}{|\Omega_1|} \sum_{p \in \Omega_1} (\beta_{1,p} \mathbf{u}_1 + \boldsymbol{\epsilon}_p) = \bar{\beta}_1 \mathbf{u}_1 + O(\sigma).
\end{equation}

The gradient update is:
\begin{equation}
  \mathbf{a}^{(t+1)} = \mathbf{a}^{(t)} + \eta \frac{\partial P}{\partial \mathbf{a}} = (c_1^{(t)} + \eta \bar{\beta}_1) \mathbf{u}_1 + c_2^{(t)} \mathbf{u}_2 + O(\eta \sigma).
\end{equation}

Projecting onto the orthonormal basis $\{\mathbf{u}_1, \mathbf{u}_2\}$ yields the claimed result.
\end{proof}

\paragraph{Interpretation.}
Theorem~\ref{thm:main} formalises the feature suppression mechanism: the gradient update selectively strengthens the dominant source representation ($c_1$ increases) while providing no learning signal for the subdominant source ($c_2$ remains constant).
This creates a self-reinforcing cycle:
\begin{enumerate}
  \item \textbf{Initial asymmetry.} By Proposition~\ref{prop:gap}, $\bar{\mathcal{S}}_1 > \bar{\mathcal{S}}_2$ due to $g_1 > g_2$.
  \item \textbf{Mask concentration.} By Lemma~\ref{lem:mask}, $\hat{m}_p$ concentrates on $\Omega_1$.
  \item \textbf{Gradient bias.} By Theorem~\ref{thm:main}, gradients flow predominantly through $\mathbf{u}_1$, increasing $c_1$.
  \item \textbf{Amplification.} As $c_1$ grows, the similarity gap $\bar{\mathcal{S}}_1 - \bar{\mathcal{S}}_2$ widens, further concentrating the mask.
  \item \textbf{Convergence.} Eventually $c_1 \gg c_2$: the model represents only the dominant source.
\end{enumerate}

\paragraph{Scope.} Our analysis explains why training collapses onto a single source once any regional gap is present, and why such a gap is self-reinforcing. It does not predict which source becomes dominant: that depends on the origin of the initial gap, which lies outside our assumptions and is left to the data. This is consistent with our framework, which relies only on the selection of one source, not on its identity.

\subsection{Connection to Simplicity Bias}

Our analysis reveals that selective convergence is an instance of the \textit{simplicity bias} phenomenon identified by Xue~\etal~\cite{xue2023icml}.
They showed that when multiple features are available for learning, gradient descent preferentially learns features with higher signal-to-noise ratios while suppressing weaker ones.
In their framework, ``easy'' features (with stronger signals) dominate the learning dynamics, causing ``hard'' features (with weaker signals) to be ignored.

Our Theorem~\ref{thm:main} provides a concrete instantiation of this principle in the multi-source audio-visual setting:
\begin{itemize}
  \item The dominant source correspondence (with the larger initial similarity $g_1$) acts as the ``easy feature'' with higher signal strength.
  \item The subdominant source correspondence (with the smaller initial similarity $g_2 < g_1$) acts as the ``hard feature'' with lower signal strength.
  \item The differentiable thresholding mechanism amplifies this bias beyond standard contrastive learning by creating a winner-take-all dynamic: once the mask concentrates on $\Omega_1$, gradients for source~2 are effectively suppressed.
\end{itemize}

Xue~\etal~\cite{xue2023icml} analysed this phenomenon in the context of class collapse and feature suppression in contrastive learning.
Our contribution is to show that the same mechanism operates in multi-source audio-visual localisation, and crucially, that we can \emph{exploit} it rather than avoid it: the resulting spatial prior $\mathbf{M}_{\text{dom}}$ enables Stage~2 to break the circular dependency and progressively reveal the subdominant source.

\section{Reproduction Details for NoPrior}
\label{sec:noprior_reproduction}

As noted in \cref{tab:two_bbox_datasets}, the original NoPrior~\cite{kim2024learning} results are reported under a source-wise evaluation protocol (see \appref{subsec:eval_protocol}), whereas we adopt a frame-wise protocol throughout our paper. Since the two protocols are not directly comparable (\appref{subsec:eval_protocol}), we reproduce NoPrior from its official codebase and evaluate it under the frame-wise protocol for a fair comparison.
Additionally, the released codebase provides runnable scripts only for the single-source setting (\texttt{concat\_num=1}); switching to the dual-source setting requires code modifications before training can proceed. Below, we document the reproduction process for both benchmarks.

\subsection{Reproduction on VGGSound-Duet}
\label{subsec:noprior_vggduet}

\paragraph{\textbf{Code modifications.}}
We base our reproduction on the official code release of Kim~\etal~\cite{kim2024learning}.
The dual-source setting (\texttt{concat\_num=2}) requires one modification to the frame concatenation in \texttt{datasets\_flow.py}.
The released code uses \texttt{torch.cat([...], dim=1)}, which concatenates two frames along the height dimension and produces a $3{\times}448{\times}224$ tensor.
However, VGGSound-Duet frames are laid out with the two sources side by side: both the original benchmark visualisations~\cite{mo2023audiovisual} and the figures in Kim~\etal~\cite{kim2024learning} show a horizontal $H{\times}2W$ layout.
We therefore change the concatenation to \texttt{dim=2} for both image and optical flow tensors, yielding a $3{\times}224{\times}448$ tensor consistent with this layout.
The model architecture, loss functions, and training logic are left entirely unchanged.
For fair comparison, we use the same synthetic dual-source training set construction as our SCAV method: pairs of single-source videos are horizontally concatenated, and their audio waveforms are mixed.

\paragraph{\textbf{Hyperparameters.}}
All hyperparameters match those reported in Kim~\etal~\cite{kim2024learning}: visual backbone is ResNet-18 (frozen), learning rate = $10^{-4}$, batch size = 128, $\alpha$ (positive threshold) = 0.65, $\omega$ (sigmoid temperature) = 0.03, $\tau_1$ (background threshold) = 0.7, $\tau_2$ (clustering threshold) = 0.6, and $\lambda_1=\lambda_2$ (loss weights)=1.0.

\paragraph{\textbf{Results.}}
The reproduced NoPrior results under the frame-wise evaluation protocol are reported in \cref{tab:two_bbox_datasets}.
\cref{fig:noprior_viz_vggduet} shows representative predictions on VGGSound-Duet.
NoPrior's Iterative Object Identification (IOI) module is designed to produce one localisation per iteration: the first iteration selects the highest-response cell and expands it into a sound-making region, then excludes that region from the similarity map $\mathbf{S}_v$ before the next iteration~\cite{kim2024learning}.
In our reproduced results, the second-iteration heatmap exhibits near-zero activation in most examples, while the first heatmap covers regions with the strongest audio-visual correspondence without separating them into distinct sources.
We note that the visualisation results in the original NoPrior supplementary material (Figure~2 in~\cite{kim2024learning}) show Object~1 and Object~2 heatmaps divided along a left/right boundary with spatially discontinuous activations. In the frame-wise evaluation setting, where two sound sources naturally co-occur within a single frame, left--right spatial position cannot serve as a prior for source separation. The released code does not contain a post-processing step that produces the spatial partition shown in Figure~2 of~\cite{kim2024learning}.

\begin{figure*}[t]
    \centering
    \includegraphics[width=\textwidth]{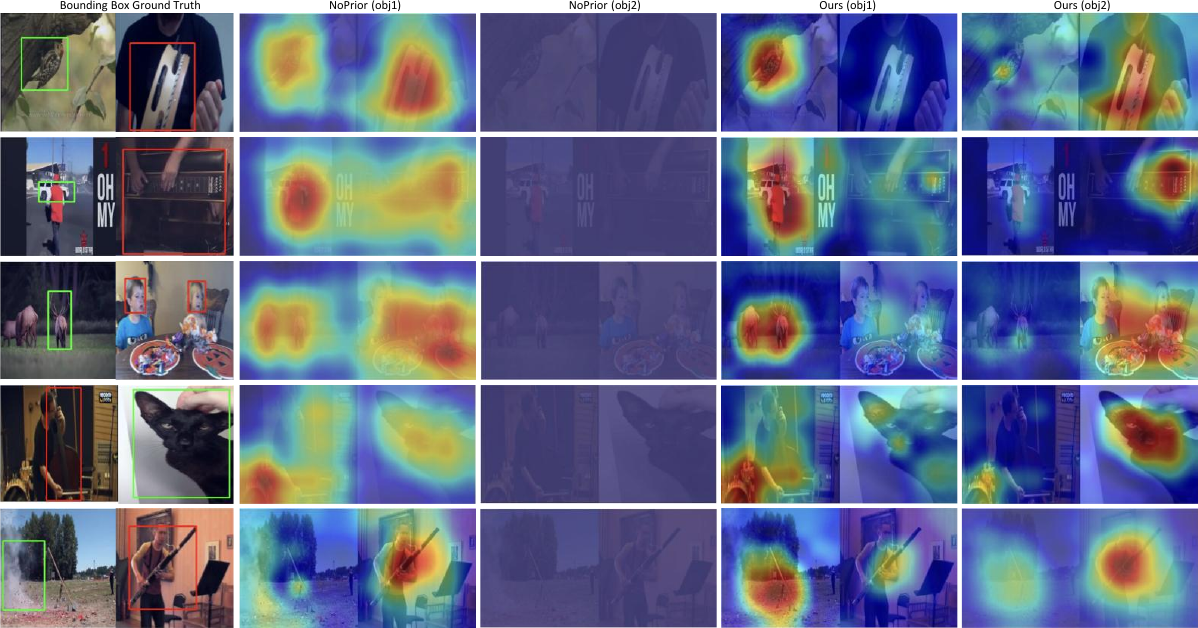}
    \caption{\textbf{Qualitative results of the reproduced NoPrior model on VGGSound-Duet.} The first column shows the input frame with ground truth annotations for both sound sources. The remaining columns display the predicted localisation heatmaps for each detected source.}
    \label{fig:noprior_viz_vggduet}
\end{figure*}

\subsection{Reproduction on MUSIC-Duet}
\label{subsec:noprior_music}

\paragraph{\textbf{Code modifications.}}
We base our reproduction on the official code release of Kim~\etal~\cite{kim2024learning}.
The dual-source setting (\texttt{concat\_num=2}) requires one modification to the frame concatenation in \texttt{datasets\_flow.py}.
The released code uses \texttt{torch.cat([...], dim=1)}, which concatenates two frames along the height dimension and produces a $3{\times}448{\times}224$ tensor.
However, VGGSound-Duet frames are laid out with the two sources side by side: both the original benchmark visualisations~\cite{mo2023audiovisual} and the figures in Kim~\etal~\cite{kim2024learning} show a horizontal $H{\times}2W$ layout.
We therefore change the concatenation to \texttt{dim=2} for both image and optical flow tensors, yielding a $3{\times}224{\times}448$ tensor consistent with this layout.
The model architecture, loss functions, and training logic are left entirely unchanged.

\paragraph{\textbf{Results.}}
The reproduced results under the frame-wise protocol are reported in \cref{tab:two_bbox_datasets}.
\cref{fig:noprior_viz_music} shows representative predictions on MUSIC-Duet, where similar behaviour to VGGSound-Duet is observed.

\begin{figure*}
    \centering
    \includegraphics[width=\textwidth]{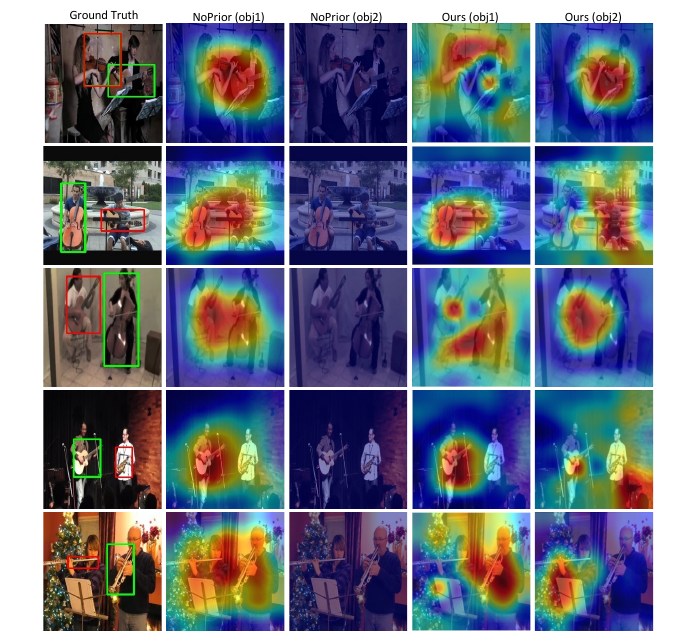}
    \caption{\textbf{Qualitative results of the reproduced NoPrior model on MUSIC-Duet.} The first column shows the input frame with ground truth annotations for both sound sources. The remaining columns display the predicted localisation heatmaps for each detected source.}
    \label{fig:noprior_viz_music}
\end{figure*}

\clearpage
%
%
\bibliographystyle{splncs04}
\bibliography{main}

\end{document}